\documentclass[10.6pt]{article}

\renewcommand{\thefootnote}{}

\usepackage{epsfig,amsmath,graphicx,amssymb,overpic}

\usepackage{setspace}
\usepackage{graphicx}
\usepackage{dcolumn}
\usepackage{bm}
\usepackage{graphicx}
\usepackage{subfigure}
\usepackage{epsfig,amsmath,graphicx,amssymb,overpic,cite}

\usepackage{graphicx}
\usepackage{dcolumn}
\usepackage{bm}
\usepackage{graphicx}
\usepackage{subfigure}
\usepackage{amsthm}

\usepackage{pdfsync}
\usepackage{cite}
\usepackage{amscd}
\usepackage{amsmath}
\usepackage{latexsym}
\usepackage{amsfonts}
\usepackage{amssymb}
\usepackage{amsthm}
\usepackage{graphicx}
\usepackage{indentfirst}
\usepackage{enumerate}
\usepackage{amstext}
\usepackage[dvipdfm,
            pdfstartview=FitH,
            CJKbookmarks=true,
            bookmarksnumbered=true,
            bookmarksopen=true,
            colorlinks=true, 
            pdfborder=001,   
            citecolor=magenta,
            linkcolor=blue,
            ]{hyperref}

 \allowdisplaybreaks[4]

\newtheorem{lemma}{Lemma}[section]
\newtheorem{propo}{Proposition}[section]
\newtheorem{prop}{RH Problem}

\newtheorem{theorem}{Theorem}[section]

\usepackage{tikz}
\usetikzlibrary{shapes,arrows,positioning}
\usepackage{xcolor} 

\def\be{\begin{equation}}
\def\ee{\end{equation}}
\def\bee{\begin{eqnarray}}
\def\ene{\end{eqnarray}}
\def\bes{\begin{subequations}}
\def\ees{\end{subequations}}

\def\d{\displaystyle}
\def\v{\vspace{0.05in}}
\def\no{{\nonumber}}

\begin{document}

\baselineskip=14pt
\renewcommand {\thefootnote}{\dag}
\renewcommand {\thefootnote}{\ddag}
\renewcommand {\thefootnote}{ }

\pagestyle{plain}

\begin{flushleft}
\baselineskip=16pt \leftline{} \vspace{-.3in} {\Large \bf  On the $N_\infty$-soliton asymptotics for the modified Camassa-Holm equation with linear dispersion \\ and  vanishing boundaries} \\[0.2in]
\end{flushleft}

\begin{flushleft}
Weifang Weng$^{\rm a}$, Zhenya Yan$^{\rm b,c,*}$
\footnote{$^{*}$Corresponding author. {\it E-mail address}: zyyan@mmrc.iss.ac.cn}
\\[0.15in]
{\footnotesize \it $^{a}$School of Mathematical Sciences, University of Electronic Science and Technology of China, Chengdu,  611731, China  \\
$^b$KLMM, Academy of Mathematics and Systems Science, Chinese Academy of Sciences, Beijing 100190, China \\
$^c$School of Mathematical Sciences, University of Chinese Academy of Sciences, Beijing 100049, China } \\
\end{flushleft}

\noindent \rule[0.25\baselineskip]{\textwidth}{0.6pt}

\noindent {\bf Abstract:}\, {\small We explore the $N_{\infty}$-soliton asymptotics for the modified Camassa-Holm (mCH) equation with linear dispersion and boundaries vanishing at infinity:
$m_t+(m(u^2-u_x^2)^2)_x+\kappa u_x=0,\quad m=u-u_{xx}$ with
$\lim_{x\rightarrow \pm \infty }u(x,t)=0$. We mainly analyze the aggregation state of $N$-soliton solutions of the mCH equation expressed by the solution of the modified Riemann-Hilbert problem in the new $(y,t)$-space when the discrete spectra are located in different regions. Starting from the modified RH problem, we find that i) when the region is a quadrature domain with $\ell=n=1$, the corresponding $N_{\infty}$-soliton is the one-soliton solution which the discrete spectral point is the center of the region; ii) when the region is a quadrature domain with $\ell=n$, the corresponding $N_{\infty}$-soliton is an $n$-soliton solution; iii) when the discrete spectra lie in the line region, we provide its corresponding Riemann-Hilbert problem,; and iv) when the discrete spectra lie in an elliptic region, it is equivalent to the case of the line region. }

\vspace{0.1in} \noindent {\small {\it MSC:} \, 35P25; 35Q15; 37K40; 35B40 }

\vspace{0.1in} \noindent {\small {\it Keywords:}\, Modified Camassa-Holm equation with linear dispersion; Inverse spectral method;
$N_\infty$-soliton asymptotics; Riemann-Hiblert problem; Discrete spectral region}

\noindent \rule[0.25\baselineskip]{\textwidth}{0.6pt}

\begin{spacing}{0.85}
\tableofcontents
\end{spacing}


\baselineskip=14pt

\section{Introduction}

In this paper, we would like to investigate the $N_\infty$-soliton asymptotics (i.e., large-$N$ asymptotics of $N$-soliton solution) for the modified Camassa-Holm (mCH) equation with  linear dispersion and zero boundary condition of the Dirichlet type at infinity~\cite{fokas1,fokas,Chen15,qiao11,Ma2}:
\begin{align}\label{mch}
\left\{\begin{aligned}
&m_t+(m(u^2-u_x^2)^2)_x+\kappa u_x=0,\quad m=u-u_{xx},\quad (x,t)\in \mathbb{R}^{2},\vspace{0.1in} \\
&\displaystyle\lim_{x\rightarrow \pm \infty }u(x,t)=0,
\end{aligned}\right.
\end{align}
where $u=u(x,t)$ is the free surface elevation indimensionless variables, and $\kappa>0$ denotes the effect of the linear dispersion, which makes the mCH equation (\ref{mch}) generate smooth soliton solutions with zero boundaries~\cite{Ma2}. Eq.~(\ref{mch}) describes the unidirectional propagation of surface waves in shallow water over a flat bottom~\cite{fokas}.
In fact, the mCH equation was originally  presented by Fokas~\cite{fokas}, then also found by Fuchssteiner~\cite{Fuch} using the symmetry method, Olver-Rosenau~\cite{or96} via the tri-Hamiltonian method, and Qiao~\cite{qiao06} via the Lax pair. Thus,
the mCH equation \eqref{mch0} was also called the Fokas-Olver-Rosenau-Qiao (FORQ) equation~\cite{Hou}.
At $\kappa\to 0$, Eq.~(\ref{mch}) reduces to the usual mCH equation~\cite{fokas,Fuch,or96,qiao06}
\bee \label{mch0}
m_t+(m(u^2-u_x^2)^2)_x=0,\quad m=u-u_{xx},
\ene
which is regarded as the modified version of the CH equation~\cite{CH,FF,CH2}
 \bee \label{ch}
  m_t + (um)_x + u_xm =0,\quad  m= u-u_{xx},
 \ene
which illustrates the unidirectional propagation of shallow water waves on a flat bottom and has a rich mathematical structure~\cite{CH2}. Recently, a Miura-type transformation was established between the mCH equation (\ref{mch0})
and CH equation (\ref{ch})~\cite{kang16}.

In 2009, Novikov~\cite{nov} used the perturbative symmetry method to classify the integrable equations
\bee
 m_t=G(u, u_x, u_{xx},...),\quad m=u-u_{xx},
 \ene
such that two integrable quadratic CH-type equations were found (i.e., the CH equation and Degasperi-Procesi equation), and two integrable cubic CH-type equations were found, that is, the usual mCH equation (\ref{mch0})
and the Novikov equation~\cite{nov}
\bee
m_t+u(m_xu+3mu_x)=0,\quad m=u-u_{xx},
\ene
The scaling transform and parameter limits can reduce the mCH equation (\ref{mch0}) to the short-pulse equation~\cite{SW}
\bee
 u_{xt}-u-\frac16(u^3)_{xx}=0.
\ene
Notice that i) the mCH equation (\ref{mch0}) with vanishing boundaries at infinity was shown to possess the non-smooth peakon solutions~\cite{qiao06,Chang}, whose stability was studied~\cite{qu1,qu2}; ii) starting from the Lax pair~\cite{qiao11}, the mCH equation (\ref{mch0}) with non-vanishing boundaries at infinity was shown to admit smooth dark solitons by the inverse scattering transform~\cite{mch-dark}; iii) the mCH equation with $\kappa\not=0$ (\ref{mch}) and vanishing boundaries as well as the mCH equation with $\kappa=0$ (\ref{mch0}) and nonvanishing boundaries at infinity were found to admit the bright multi-smooth solutions via the reciprocal transform and bilinear method, respectively~\cite{Ma1,Ma2}.

The Riemann-Hilbert (RH) problem with the inverse scattering transform (IST)~\cite{GGKM} plays a more and more important role in the study of integrable systems~\cite{Zhou89,Its03,Zhou03,Olver16}. In fact, in 1974, the RH problem was first used by Zakharov-Shabat~\cite{zs1} to solve integrable systems based on the IST.
After that, the IST and/or RH problem can be used to find not only the exact solitons of integrable equations for the case of reflectioness potential~\cite{book-F,zs2,Ablowitz2007, Yang10, soliton-84, RIST}, but also the large-order asymptotic behaviors of solutions  of integrable equations for the case of reflection potential or solitonness and other types of asymptotics of solitons~\cite{bil1,bil2,bil3}. On the one hand, in 1993, Deift-Zhou~\cite{DF} proposed a nonlinear steepest descent method to find the long-time asymptotics behavior of the solution for the mKdV equation in terms of modified RH problems. Recently, McLaughlin {\it et al}~\cite{Dbar1,Dbar2} extended the Deift-Zhou steepest descent method to present the $\bar\partial$-steepest descent method to study asymptotic behaviors of both orthogonal polynomials with non-analytical weights~\cite{Dbar1,Dbar2}, and integrable systems~\cite{Dbar3,NLS-18}. There were other results on long-time asymptotic behaviors of solutions of some integrable systems (see, e.g., Refs.~\cite{Bio,BI, GT,  TV, B1,B2, RS, huang,fan22,fan-adv,BL1, BL2, LE, CL} and references therein).

On the other hand, in 1971, Zahkarov~\cite{zak71} first analyzed the large $N$-limit of the $N$-soliton solution of the KdV equation, which is called the $N_{\infty}$-soliton behavior or siloton gas. Afterwards, the study of soliton gas will be extended to investigate the fluid dynamics of soliton gas, breather gas, dense soliton gas for other nonlinear wave equations, such as the NLS equation,  KdV equation, modified KdV equation, etc.~\cite{E1,E2,E3,E5,E6,E4,E7,Girotti-1,Girotti-2,Grava-3} by using the numerical method and RH problem.

 Without loss of generality, one can choose $\kappa=2$ in the mCH equation (\ref{mch}) via the scaling transforms:
$u(x,t)=\left(\kappa/2\right)^{\frac12}\tilde{u}(\tilde{x},\tilde{t}),\quad x=\tilde{x} ,\quad t=2\tilde{t}/\kappa.$
 The mCH equation (\ref{mch}) with non-zero linear dispersive coefficient $\kappa=2$ is completely integrable and possesses the Lax pair~\cite{qiao06,sch96,xu19}:
\begin{align}\label{lax}
\begin{aligned}
&\varPhi_x=X\varPhi, \quad X=-\frac{k(z)}{2}\sigma_3+\frac{i\lambda(z) m}{2}\sigma_2,\quad m=u-u_{xx},\v\\
&\varPhi_t=T\varPhi, \quad T=k(z)\left(\frac{1}{\lambda^2(z)}+\frac{u^2-u_x^2}{2}\right)\sigma_3
-i\left(\frac{u-k(z)u_x}{\lambda(z)}+\frac{\lambda(z)(u^2-u_x^2)m}{2}\right)\sigma_2,
\end{aligned}
\end{align}
where the three Pauli matrices are
\bee\no
&\sigma_1=\!\left[\!\!\begin{array}{cc}
0& 1  \\
1 & 0
\end{array}\!\!\right],\quad \sigma_2=\!\left[\!\!\begin{array}{cc}
0& -i  \\
i & 0
\end{array}\!\!\right],
\quad\sigma_3=\!\left[\!\!\begin{array}{cc}
1& 0  \\
0 & -1
\end{array}\!\!\right].
\ene
and
\bee\label{kl}
k(z)=\frac{i}{2}(z-\frac1z),\quad  \lambda(z)=\frac{1}{2}(z+\frac1z),
\ene
$z\in\mathbb{C}$ is a spectral parameter.
In 2020, Boutet de Monvel-Karpenko-Shepelsky~\cite{B1,B2} first presented the RH problem of the mCH equation (\ref{mch}) with nonzero backgrounds and  gave the long-time asymptotics of solution for the solitonless case. Recently,
the long-time asymptotic behaviors were found for the Eq.~(\ref{mch}) with Schwartz initial data~\cite{xu19} and weighted Sobolev initial data~\cite{fan-adv} by using the Deift-Zhou steepest decedent method and $\bar\partial$-steepest decedent method, respectively. More recently,  the solutions of the mCH equation (\ref{mch0}) with step-like initial data were found via the solution of the RH problem~\cite{kst}, and then its long-time asymptotics was studied~\cite{fan22}.
However, to the best of our knowledge,
there was no report on the $N_\infty$-soliton asymptotics of the mCH equation with linear dispersion \eqref{mch} before.

In this paper we will, motivated by the idea for the NLS equation~\cite{Grava-3}, investigate the $N_{\infty}$-soliton behaviors (i.e., the large-$N$ asymptotics of $N$-soliton solution) for the mCH equation with  linear dispersion and zero boundary condition of the Dirichlet type at infinity given by Eq.~(\ref{mch}). Different from the works on soliton gases for these NLS, mKdV and KdV equations~\cite{Girotti-1,Girotti-2,Grava-3}, the Lax pair (\ref{lax}) of the mCH equation with linear dispersion (\ref{mch})  has more singularities at $k=\infty,\, \lambda=0, \infty$, i.e., $z=0,\, \infty$ and $z=\pm i$ (branch
cut points in the complex $z$-plane), which lead to more asymptotic behaviors of the Lax pair (\ref{lax}). This difficult and key point was solved by Boutet de Monvel {\it et al}~\cite{B1,B2,BS1,BS2,BS3} with the aid of an appropriate transform.
We study the aggregation state of $N_{\infty}$-soliton solutions of Eq.~\eqref{mch} when the infinite many discrete spectra are located in different regions. We find that when the region is a quadrature domain and $\ell=n=1$, the corresponding $N_{\infty}$-soliton becomes the one-soliton solution, where the discrete spectral point is the center of the region. We show that when the region is a quadrature domain and $\ell=n$, the corresponding $N_{\infty}$-soliton is equivalent to an $n$-soliton solution. When the discrete spectra lie in the line region, we provide its corresponding Riemann-Hilbert problem. When the discrete spectra lie in an elliptic region, it is equivalent to the case of the line region.

The main contributions of this paper about the $N_{\infty}$-soliton asymptotics of the modified Camassa-Holm equation are listed follows:

\begin{itemize}

\item [1)] The discrete spectra of the mCH equation (\ref{mch}) possess multiple symmetry conditions due to its Lax pair (\ref{lax}) possessing more singularities at $k=\infty,\, \lambda=0, \infty$, i.e., $z=0,\, \infty$ and $z=\pm i$ (branch
cut points in the complex $z$-plane), which complicate the analysis of their residues in the RH problem. When the discrete spectra are uniformly distributed across the different regions, such as quadrature and elliptic regions, we can find the aggregation states of $N$-soliton solutions.

\item [2)] When the discrete spectra of the mCH equation (\ref{mch}) are uniformly distributed on a line, the corresponding RH problem becomes intricate due to the multiple regularity conditions and residue conditions. Through a series of deformations of the RH problem, we obtain a solvable RH problem corresponding to this situation.

\item [3)] The solvable RH problem presented in the Point 2), despite being a step forward, remains complex. We propose for the first time the high-order functions $f(z),\, g(z)$, which are used to reduce this RH problem to a standard form. By analyzing its asymptotic properties on the arcs and at the endpoints, we can derive the $N_{\infty}$-soliton asymptotics for
    the mCH equation (\ref{mch}).

\end{itemize}

The rest of this paper is arranged as follows. In Sec. 2, we recall the basic RH problem of the mCH equation with linear dispersion (\ref{mch}), used to solve its smooth $N$-soliton solutions~\cite{B1,B2,xu19,fan-adv}. In Sec. 3, we introduce the modified RH problem with the jump curves only being some very small radius encircling the discrete spectra,  which is used to construct the $N$-soliton solutions of the mCH equation (\ref{mch}). In Secs. 4, 5 and 6, we, based on some modified RH problems, investigate
the $N_{\infty}$-soliton asymptotics for the mCH equation with linear dispersion in the different types of domains for the discrete spectra. Finally, we give the conclusions and discussions in Sec. 7.

\section{Preliminaries}

In this section, we recall some main results on the RH problem generated from its Lax pair (\ref{lax}) such that the RH problem can be used to construct the $N$-soliton solutions of the mCH equation (\ref{mch}) with linear dispersion~\cite{B1,B2,xu19,fan-adv}.
The Lax pair (\ref{lax}) for the mCH equation has singularities at $k=\infty,\, \lambda=0, \infty$, i.e., $z=0,\, \infty,\, \pm i$ (cf. Eq.~(\ref{kl})), where $z=\pm i$ is also called the branch cut points.

\subsection{The modified Jost solutions}

{\it Case 1.} As $z\to \infty$, let
\bee\label{mu-1}
\mu_{\pm}=\frac{2q(q+1)^{3/2}}{m^2+(q+1)^2}\left(\begin{array}{cc} 1&  \frac{im}{q+1} \v\\
\frac{im}{q+1} & 1 \end{array}\right)\varPhi_{\pm}e^{\frac{i}{4}(z-1/z)p(x,t;z)\sigma_3}, \quad x\to \pm\infty
\ene
with
\bee
p(x,t;z)=x-\int_x^{\infty}(q-1)dy-8(z+1/z)^{-2}t,\quad q=\sqrt{m+1}.
\ene
Then the Lax pair (\ref{lax}) reduces to one for $\mu_{\pm}$:
\bee \label{lax2}
\left\{\begin{array}{l}
 (\mu_{\pm})_x=\dfrac{i}{4}(z^{-1}-z)p_x[\sigma_3, \mu_{\pm}]+U\mu_{\pm}, \v\\
 (\mu_{\pm})_t=\dfrac{i}{4}(z^{-1}-z)p_t[\sigma_3, \mu_{\pm}]+V\mu_{\pm}.
\end{array} \right.
\ene
with
\bee
\begin{array}{l}
\d U=\frac{im_x}{2q^2}\sigma_1-\frac{im}{2zq}(m\sigma_3-\sigma_2), \v\\
\d V=\left(\frac{im_t}{2q^2}+\frac{(z^2-1)u_x}{z^2+1}\right)\sigma_1+\frac{im(u^2-u_x^2)}{2zq}(m\sigma_3-\sigma_2).
\end{array}
\ene
It follows from the Lax pair (\ref{lax2}) that the two Volterra-type integrals can be written as
\bee
\mu_{\pm}=I+\int^x_{\pm\infty}e^{\frac{i}{4}(z^{-1}-z)(p(x)-p(y))\widehat{\sigma}_3}U(y)\mu_{\pm}(y)dy,\quad
e^{\alpha\widehat{\sigma}_3}A=e^{\alpha\sigma_3}Ae^{-\alpha\sigma_3}.
\ene

There exists the scattering matrix $S(z)$ depending on only the spectral parameter $z$ between $\mu_{\pm}$:
\bee
\mu_-=\mu_+e^{\frac{i}{4}(z^{-1}-z)(p(x)-p(y))\widehat{\sigma}_3}S(z)
\quad
S(z)=\left(\begin{array}{cc} a(z) & -b^*(z^*) \v\\ b(z) & a^*(z^*) \end{array}\right)
\ene
with the symmetries $S(z)=S^*(1/z^*)=\sigma_3S(-1/z)\sigma_3$. As a result, one has the symmetries of
the reflection coefficient: $\rho(z)=\rho^*(1/z^*)=-\rho^*(-z^*)=\rho(-1/z^)$, where $\rho(z)=b(z)/a(z)$ and the star denotes the
complex conjugate.

Let $\{z_n,\ -z_n^*,\, 1/z_n,\, -1/z_n^*\}$ with $|z_n|>1$ and ${\rm arg}(z_n)\in (0, \pi/2]$ and
$\{w_n,\, -w_n^*\}$ with  $|w_n|=1$ and ${\rm arg}(w_n)\in (0, \pi/2]$  be the simple zeros of $a(z)$. Then
the discrete spectra is $Z\cup Z^*$ with $Z=K\cup W,\, Z^*=K^*\cup W^*$, $K=\{z_n,\ -z_n^*,\, 1/z_n,\, -1/z_n^*\}_{n=1}^{N},\,
W=\{w_n,\ -w_n^*\}_{n=1}^{N'}$.

\v {\it Case 2.} As $z\to 0$, similar to the case $z\to \infty$, one has $a(z)\to 1,\, b(z)\to 0$ as $z\to 0$.

\v {\it Case 3.} As $z\to \pm i$ (i.e., $\lambda(z)\to 0$),  let
\bee\label{mu-0}
\mu_{\pm}^{(0)}=\varPhi_{\pm}e^{(\frac{k}{2}x-\frac{k}{\lambda^2}t)\sigma_3}.
\ene
Then the Lax pair (\ref{lax}) reduces to
\bee \label{lax3}
\begin{array}{l}
 (\mu_{\pm}^{(0)})_x=-\dfrac{k}{2}[\sigma_3, \mu_{\pm}^{(0)}]+U_0\mu_{\pm}^{(0)}, \v\\
 (\mu_{\pm}^{(0)})_t=\dfrac{k}{\lambda^2}[\sigma_3, \mu_{\pm}^{(0)}]+V_0\mu_{\pm}^{(0)},
\end{array}
\ene
with
\bee
U_0=\frac{i}{2}\lambda m\sigma_2,\qquad
V_0=\frac{ku_x}{\lambda}\sigma_1+\frac{u^2-u_x^2}{2}(k\sigma_3-i\lambda m\sigma_2)
-\frac{iu}{\lambda}\sigma_2.
\ene
One has the asymptotic expansion
\bee
 \mu^{(0)}=I+ \mu_1^{(0)}(z-i)+O((z-i)^2),\quad z\to i.
\ene
It follows from the Lax pair (\ref{lax3}) that two Volterra type integrals can be written as
\bee
\mu_{\pm}^{(0)}=I+\int^x_{\pm\infty}e^{-\frac{k}{2}(x-y)\widehat{\sigma}_3}U_0(y)\mu_{\pm}^{(0)}(y)dy.
\ene

 It follows from Eqs.~(\ref{mu-1}) and (\ref{mu-0}) that one has
 \bee
\mu_{\pm}=\frac{2q(q+1)^{3/2}}{m^2+(q+1)^2}\left(\begin{array}{cc} 1&  \frac{im}{q+1} \v\\
\frac{im}{q+1} & 1 \end{array}\right)\mu_{\pm}^{(0)}e^{\frac{i}{4}(z-1/z)h_{\pm}(x,t)\sigma_3},\quad
h_{\pm}=\int_{\pm\infty}^x(q-1)dy.
 \ene

\subsection{The basic Riemann-Hilbert problem and $N$-soliton solution}

To construct the RH problem, let
\bee
 y(x,t)=x-\int_{x}^{+\infty}(q(s)-1)ds,\quad s=x-h_+(x,t),
\ene
and a piecewise meromorphic function be
\bee\label{RHP-M}
M^{(0)}(y, t; z)=\left\{
\begin{array}{ll}
\left(\dfrac{\mu_{-1}(x(y,t), t; z)}{a(z)},\, \mu_{+2}(x(y,t), t; z)\right), & z\in \mathbb{C}^+, \v\v \\
\left(\mu_{+1}(x(y,t), t; z),\, \dfrac{\mu_{-2}(x(y,t), t; z)}{a^*(z^*)}\right),  & z\in \mathbb{C}^-.
\end{array}\right.
\ene
Then, the matrix function $M^{(0)}(y,t;z)$ satisfies the following RH problem:

\begin{prop}\label{RH1}
Find a $2\times 2$ matrix $M^{(0)}(y,t;z)$ that satisfies the following conditions:

\begin{itemize}

 \item {} Analyticity: $M^{(0)}$ is meromorphic in $\{z|z\in\mathbb{C}\setminus \mathbb{R}\}$ and takes
continuous boundary values on $\mathbb{R}$;

 \item {} The jump condition: the boundary values on the jump contour $\Sigma$ are defined as
 \bee
  M_+^{(0)}(z)=M_-^{(0)}(z)J(z), \quad J(z)=\mathrm{e}^{i\theta(x(y,t), t; z)\widehat\sigma_3}\left(\begin{array}{cc} 1+|\rho(z)|^2 & \widehat\rho(z)\\[0.05in] \rho(z)& 1 \end{array}\right),\,\,\, k\in\mathbb{R},
  \ene
where $\theta(x(y,t), t; z)=\frac{i}{2}k(z)\left(y-2\lambda^{-2}(z)t\right)$.

 \item {} Normalization:
\bee
 M^{(0)}=\left\{\begin{array}{ll}
    \mathbb{I}_2+O\left(1/z\right),  & z\to\infty, \v\\
    \dfrac{(q+1)^2}{m^2+(q+1)^2}\left(\begin{array}{cc} 1&  \frac{im}{q+1} \v\\
\frac{im}{q+1} & 1 \end{array}\right)\left[\mathbb{I}_2+\mu_1^{(0)}(z-i)\right]e^{\frac12h_+\sigma_3}+\mathcal{O}((z-i)^2), & z\to i.
\end{array}\right.
\ene

\item {} Residue conditions: $M^{(0)}$ has simple poles at each point in $K:=\{z_n,\ -z_n^*,\, \frac{1}{z_n},\, -\frac{1}{z_n^*}\}_{j=1}^N$ with:
\begin{eqnarray}
&&\no \mathop{\mathrm{Res}}_{z=z_j}M^{(0)}(y, t; z)=\lim\limits_{z\to z_j}M^{(0)}(y, t; z)\left[\!\!\begin{array}{cc}
0& 0  \vspace{0.05in}\\
c_je^{-2i\theta(z_j)}& 0
\end{array}\!\!\right],\vspace{0.05in}\\
&&\no
\mathop{\mathrm{Res}}_{z=-z_j^*}M^{(0)}(y, t; z)=\lim\limits_{z\to -z_j^*}M^{(0)}(y, t; z)\left[\!\!\begin{array}{cc}
0& 0  \vspace{0.05in}\\
c_j^*e^{-2i\theta(-z_j^*)}& 0
\end{array}\!\!\right],\vspace{0.05in}\\
&&\no \mathop{\mathrm{Res}}_{z=\frac{1}{z_j^*}}M^{(0)}(y,t;z)=\lim\limits_{z\to \frac{1}{z_j^*}}M^{(0)}(y,t;z)\left[\!\!\begin{array}{cc}
0& 0  \vspace{0.05in}\\
-\frac{c_j^*}{z_j^{*2}}e^{-2i\theta(\frac{1}{z_j^*})}& 0
\end{array}\!\!\right],\vspace{0.05in}\\
&&\no \mathop{\mathrm{Res}}_{z=-\frac{1}{z_j}}M^{(0)}(y,t;z)=\lim\limits_{z\to -\frac{1}{z_j}}M^{(0)}(y,t;z)\left[\!\!\begin{array}{cc}
0& 0  \vspace{0.05in}\\
-\frac{c_j}{z_j^{2}}e^{-2i\theta(-\frac{1}{z_j})}& 0
\end{array}\!\!\right],\vspace{0.05in}\\
&&\no \mathop{\mathrm{Res}}_{z=z_j^*}M^{(0)}(y,t;z)=\lim\limits_{z\to z_j^*}M^{(0)}(y,t;z)\left[\!\!\begin{array}{cc}
0& -c_j^*e^{2i\theta(z_j^*)}  \vspace{0.05in}\\
0& 0
\end{array}\!\!\right], \vspace{0.05in}\\
&&\no \mathop{\mathrm{Res}}_{z=-z_j}M^{(0)}(y,t;z)=\lim\limits_{z\to -z_j}M^{(0)}(y,t;z)\left[\!\!\begin{array}{cc}
0& -c_je^{2i\theta(-z_j)}  \vspace{0.05in}\\
0& 0
\end{array}\!\!\right], \vspace{0.05in}\\
&&\no \mathop{\mathrm{Res}}_{z=\frac{1}{z_j}}M^{(0)}(y,t;z)=\lim\limits_{z\to \frac{1}{z_j}}M^{(0)}(y,t;z)\left[\!\!\begin{array}{cc}
0& \frac{c_j}{z_j^2}e^{2i\theta(\frac{1}{z_j})}  \vspace{0.05in}\\
0& 0
\end{array}\!\!\right], \vspace{0.05in}\\
&&\no \mathop{\mathrm{Res}}_{z=-\frac{1}{z_j^*}}M^{(0)}(y,t;z)=\lim\limits_{z\to -\frac{1}{z_j^*}}M^{(0)}(y,t;z)\left[\!\!\begin{array}{cc}
0& \frac{c_j^*}{z_j^{*2}}e^{2i\theta(-\frac{1}{z_j^*})}  \vspace{0.05in}\\
0& 0 \end{array}\!\!\right].
\end{eqnarray}
and at each point in $W:=\{w_j,-w_j^*\}_{j=1}^{N_2}$ with:
\begin{eqnarray}
&& \no
\mathop{\mathrm{Res}}_{z=w_j}M^{(0)}(y,t;z)=\lim\limits_{z\to w_j}M^{(0)}(y,t;z)\left[\!\!\begin{array}{cc}
0& 0  \vspace{0.05in}\\ d_je^{-2i\theta(w_j)}& 0 \end{array}\!\!\right],\vspace{0.05in}\\
&& \no  \mathop{\mathrm{Res}}_{z=-w_j^*}M^{(0)}(y,t;z)=\lim\limits_{z\to -w_j^*}M^{(0)}(y,t;z)\left[\!\!\begin{array}{cc}
0& 0  \vspace{0.05in}\\ d_j^*e^{-2i\theta(-w_j^*)}& 0 \end{array}\!\!\right],\vspace{0.05in}\\
&& \no \mathop{\mathrm{Res}}_{z=w_j^*}M^{(0)}(y,t;z)=\lim\limits_{z\to w_j^*}M^{(0)}(y,t;z)\left[\!\!\begin{array}{cc}
0& -e_j^*e^{2i\theta(w_j^*)}  \vspace{0.05in}\\ 0& 0 \end{array}\!\!\right], \vspace{0.05in}\\
&& \no  \mathop{\mathrm{Res}}_{z=-w_j}M^{(0)}(y,t;z)=\lim\limits_{z\to -w_j}M^{(0)}(y,t;z)\left[\!\!\begin{array}{cc}
0& -d_je^{2i\theta(-w_j)}  \vspace{0.05in}\\ 0& 0 \end{array}\!\!\right]
\end{eqnarray}
with $c_j$'s and $d_j$'s being complex constants.
\end{itemize}
\end{prop}
Then the $N$-soliton solution $u(x,t)$ of the mCH equation (\ref{mch}) is given by
\bee\label{fanyan}
u(x,t)=\lim\limits_{z\rightarrow i}\frac{1}{z-i}\left(1-\dfrac{(M^{(0)}_{11}(z)+M^{(0)}_{21}(z))(M^{(0)}_{12}(z)+M^{(0)}_{22}(z))}
{(M^{(0)}_{11}(i)+M^{(0)}_{21}(i))(M^{(0)}_{12}(i)+M^{(0)}_{22}(i))}\right),
\ene
where
\bee\no
x(y,t)=y+h_+(x,t)=y-\ln\left(\dfrac{M^{(0)}_{12}(i)+M^{(0)}_{22}(i)}{M^{(0)}_{11}(i)+M^{(0)}_{21}(i)}\right).
\ene

\section{The modified Riemann-Hilbert problem}

To study the $N_{\infty}$-soliton asymptotic behaviors of the mCH equation with linear dispersion (\ref{mch}) for the discrete spectra $K\cup K^*$ considered in this paper, we define a closed curve $\Gamma_{1+}\,(\Gamma_{2+},\,\Gamma_{3+},\,\Gamma_{4+})$ with a very small radius encircling the simple poles $\{z_j\}_{j=1}^N\,(\{-z_j^*\}_{j=1}^N,\,\{\frac{1}{z_j^*}\}_{j=1}^N,\, \{-\frac{1}{z_j}\}_{j=1}^N)$ counterclockwise in the upper half plane $\mathbb{C}_+$, respectively, and a closed curve $\Gamma_{1-}\,(\Gamma_{2-},\, \Gamma_{3-},\,\Gamma_{4-})$ with a very small radius encircling the poles $\{z_j^*\}_{j=1}^N\, (\{-z_j\}_{j=1}^N,\, \{\frac{1}{z_j}\}_{j=1}^N,\, \{-\frac{1}{z_j^*}\}_{j=1}^N)$ counterclockwise in the lower half plane $\mathbb{C}_-$, respectively. Then we make the following transform for  $M^{(0)}(y,t;z)$:
\bee
M^{(1)}(y,t;z)
=\begin{cases}M^{(0)}(y,t;z)\left[\!\!\begin{array}{cc}
1& 0  \vspace{0.05in}\\
-\sum\limits_{j=1}^{N}\dfrac{c_je^{-2i\theta(z_j)}}{z-z_j}& 1
\end{array}\!\!\right],\quad z~\mathrm{within}~\Gamma_{1+},\vspace{0.05in}\\
M^{(0)}(y,t;z)\left[\!\!\begin{array}{cc}
1& 0  \vspace{0.05in}\\
-\sum\limits_{j=1}^{N}\dfrac{c_j^*e^{-2i\theta(-z_j^*)}}{z+z_j^*}& 1
\end{array}\!\!\right],\quad z~\mathrm{within}~\Gamma_{2+},\vspace{0.05in}\\
M^{(0)}(y,t;z)\left[\!\!\begin{array}{cc}
1& 0  \vspace{0.05in}\\
\sum\limits_{j=1}^{N}\dfrac{\frac{c_j^*}{z_j^{*2}}e^{-2i\theta(\frac{1}{z_j^*})}}{z-\frac{1}{z_j^*}}& 1
\end{array}\!\!\right],\quad z~\mathrm{within}~\Gamma_{3+},\vspace{0.05in}\\
M^{(0)}(y,t;z)\left[\!\!\begin{array}{cc}
1& 0  \vspace{0.05in}\\
\sum\limits_{j=1}^{N}\dfrac{\frac{c_j}{z_j^{2}}e^{-2i\theta(-\frac{1}{z_j})}}{z+\frac{1}{z_j}}& 1
\end{array}\!\!\right],\quad z~\mathrm{within}~\Gamma_{4+},
\end{cases}
\ene
\bee
M^{(1)}(y,t;z)
=\begin{cases}
M^{(0)}(y,t;z)\left[\!\!\begin{array}{cc}
1& \sum\limits_{j=1}^{N}\dfrac{c_j^*e^{2i\theta(z_j^*)}}{z-z_j^*}  \vspace{0.05in}\\
0& 1
\end{array}\!\!\right],\quad z~\mathrm{within}~\Gamma_{1-},\v\\
M^{(0)}(y,t;z)\left[\!\!\begin{array}{cc}
1&  \sum\limits_{j=1}^{N}\dfrac{c_je^{2i\theta(-z_j)}}{z+z_j} \vspace{0.05in}\\
0& 1
\end{array}\!\!\right],\quad z~\mathrm{within}~\Gamma_{2-},\v\\
M^{(0)}(y,t;z)\left[\!\!\begin{array}{cc}
1& -\sum\limits_{j=1}^{N}\dfrac{\frac{c_j}{z_j^2}e^{2i\theta(\frac{1}{z_j})}}{z-\frac{1}{z_j}}  \vspace{0.05in}\\
0& 1
\end{array}\!\!\right],\quad z~\mathrm{within}~\Gamma_{3-},\v\\
M^{(0)}(y,t;z)\left[\!\!\begin{array}{cc}
1& -\sum\limits_{j=1}^{N}\dfrac{\frac{c_j^*}{z_j^{*2}}e^{2i\theta(-\frac{1}{z_j^*})}}{z+\frac{1}{z_j^*}} \vspace{0.05in}\\
0& 1
\end{array}\!\!\right],\quad z~\mathrm{within}~\Gamma_{4-},\v\\
M^{(0)}(y,t;z),\quad \mathrm{otherwise}.
\end{cases}
\ene

Then matrix function $M^{(1)}(y,t;z)$ satisfies the following Riemann-Hilbert problem.

\begin{prop}\label{RH2}
Find a $2\times 2$ matrix function $M^{(1)}(y,t;z)$ that satisfies:

\begin{itemize}

 \item {} Analyticity: $M^{(1)}(y,t;z)$ is analytic in $\mathbb{C}\setminus(\Gamma_{1\pm}\cup\Gamma_{2\pm}\cup\Gamma_{3\pm}\cup\Gamma_{4\pm})$ and takes continuous boundary values on $\Gamma_{1\pm}\cup\Gamma_{2\pm}\cup\Gamma_{3\pm}\cup\Gamma_{4\pm}$.

 \item {} Jump condition: The boundary values on the jump contour $\Gamma_{1\pm}\cup\Gamma_{2\pm}\cup\Gamma_{3\pm}\cup\Gamma_{4\pm}$ are defined as
 \bee
M^{(1)}_{+}(y,t;z)=M^{(1)}_{-}(y,t;z)V_1(y,t;z),\quad z\in\Gamma_{1\pm}\cup\Gamma_{2\pm}\cup\Gamma_{3\pm}\cup\Gamma_{4\pm},
\ene
where
\bee\label{V1-1}
V_1(y,t;z)
=\begin{cases}
\left[\!\!\begin{array}{cc}
1& 0  \vspace{0.05in}\\
-\sum\limits_{j=1}^{N}\dfrac{c_je^{-2i\theta(z_j)}}{z-z_j}& 1
\end{array}\!\!\right],\quad z\in\Gamma_{1+},\vspace{0.05in}\\
\left[\!\!\begin{array}{cc}
1& 0  \vspace{0.05in}\\
-\sum\limits_{j=1}^{N}\dfrac{c_j^*e^{-2i\theta(-z_j^*)}}{z+z_j^*}& 1
\end{array}\!\!\right],\quad z\in\Gamma_{2+},\vspace{0.05in}\\
\left[\!\!\begin{array}{cc}
1& 0  \vspace{0.05in}\\
\sum\limits_{j=1}^{N}\dfrac{\frac{c_j^*}{z_j^{*2}}e^{-2i\theta(\frac{1}{z_j^*})}}{z-\frac{1}{z_j^*}}& 1
\end{array}\!\!\right],\quad z\in\Gamma_{3+},\vspace{0.05in}\\
\left[\!\!\begin{array}{cc}
1& 0  \vspace{0.05in}\\
\sum\limits_{j=1}^{N}\dfrac{\frac{c_j}{z_j^{2}}e^{-2i\theta(-\frac{1}{z_j})}}{z+\frac{1}{z_j}}& 1
\end{array}\!\!\right],\quad z\in\Gamma_{4+},\vspace{0.05in}\\
\end{cases}
\ene
\bee\label{V1-1g}
V_1(y,t;z)
=\begin{cases}
\left[\!\!\begin{array}{cc}
1& \sum\limits_{j=1}^{N}\dfrac{c_j^*e^{2i\theta(z_j^*)}}{z-z_j^*}  \vspace{0.05in}\\
0& 1
\end{array}\!\!\right],\quad z\in\Gamma_{1-},\v\\
\left[\!\!\begin{array}{cc}
1&  \sum\limits_{j=1}^{N}\dfrac{c_je^{2i\theta(-z_j)}}{z+z_j} \vspace{0.05in}\\
0& 1
\end{array}\!\!\right],\quad z\in\Gamma_{2-},\v\\
\left[\!\!\begin{array}{cc}
1& -\sum\limits_{j=1}^{N}\dfrac{\frac{c_j}{z_j^2}e^{2i\theta(\frac{1}{z_j})}}{z-\frac{1}{z_j}}  \vspace{0.05in}\\
0& 1
\end{array}\!\!\right],\quad z\in\Gamma_{3-},\v\\
\left[\!\!\begin{array}{cc}
1& -\sum\limits_{j=1}^{N}\dfrac{\frac{c_j^*}{z_j^{*2}}e^{2i\theta(-\frac{1}{z_j^*})}}{z+\frac{1}{z_j^*}} \vspace{0.05in}\\
0& 1
\end{array}\!\!\right],\quad z\in\Gamma_{4-}.
\end{cases}
\ene

 \item {} Normalization:
\bee
 M^{(1)}=\left\{\begin{array}{ll}
    \mathbb{I}_2+O\left(1/z\right),  & z\to\infty, \v\\
    \dfrac{(q+1)^2}{m^2+(q+1)^2}\left(\begin{array}{cc} 1&  \frac{im}{q+1} \v\\
\frac{im}{q+1} & 1 \end{array}\right)\left[\mathbb{I}_2+\mu_1^{(0)}(z-i)\right]e^{\frac12h_+\sigma_3}+\mathcal{O}((z-i)^2), & z\to i.
\end{array}\right.
\ene

\end{itemize}
\end{prop}

According to Eq.~(\ref{fanyan}), we can recover the $N$-soliton solution $u(x,t)$ of the mCH equation (\ref{mch}) by the following formula:
\bee\label{fanyan-1}
u(x,t)=\lim\limits_{z\rightarrow i}\frac{1}{z-i}\left(1-\dfrac{(M^{(1)}_{11}(z)+M^{(1)}_{21}(z))(M^{(1)}_{12}(z)+M^{(1)}_{22}(z))}{(M^{(1)}_{11}(i)+M^{(1)}_{21}(i))(M^{(1)}_{12}(i)+M^{(1)}_{22}(i))}\right).
\ene
where
\bee\no
x(y,t)=y-\ln\left(\dfrac{M^{(1)}_{12}(i)+M^{(1)}_{22}(i)}{M^{(1)}_{11}(i)+M^{(1)}_{21}(i)}\right).
\ene

In the following, we consider the $N_{\infty}$-soliton asymptotic properties of $u(x,t)$ defined by Eq.~(\ref{fanyan-1}) when the discrete spectra are located in different regions, such as quadrature domain, line domain, and elliptic domain.

\section{$N_{\infty}$-soliton asymptotics: the quadrature domain}

In this section, we care about the $N_{\infty}$ asymptotic situation of $N$-soliton solution  under the additional assumptions:
\begin{itemize}

 \item {} The discrete spectra $z_j,j=1,\cdots,N$ fill uniformly compact domain $\Omega_1$ which is strictly contained in the domain $D_{\Gamma_{1+}}$ bounded by $\Gamma_{1+}$, that is,
\bee
\Omega_1:=\{z|~|(z-s_1)^\ell-s_2|<s_3\},\quad \Omega_1\subset D_1:=\{z\in\mathbb{C}|~0<\mathrm{arg}z<\frac{\pi}{2},|z|>1\},
\ene
where $\ell\in\mathbb{N}^+,s_1\in\mathbb{C}^+$ and $|s_2|,s_3$ are sufficiently small (see Fig.~\ref{fig1}).

 \item {} The norming constants $c_j,j=1,\cdots,N$ have the following form:
\bee\label{c-j}
c_j=\frac{|\Omega_1|r(z_j,z_j^*)}{N\pi}.
\ene
where $|\Omega_1|$ means the area of the domain $\Omega_1$ and $r(z,z^*):=n(z^*-s_1^*)^{n-1}r_1(z)$ is a smooth function with respect to variables $z$ and $z^*$ and the function $r_1(z)$ is analytic in the domain $\Omega_1$.

\end{itemize}

\begin{figure}[!t]
    \centering
 \vspace{-0.15in}
  {\scalebox{0.38}[0.38]{\includegraphics{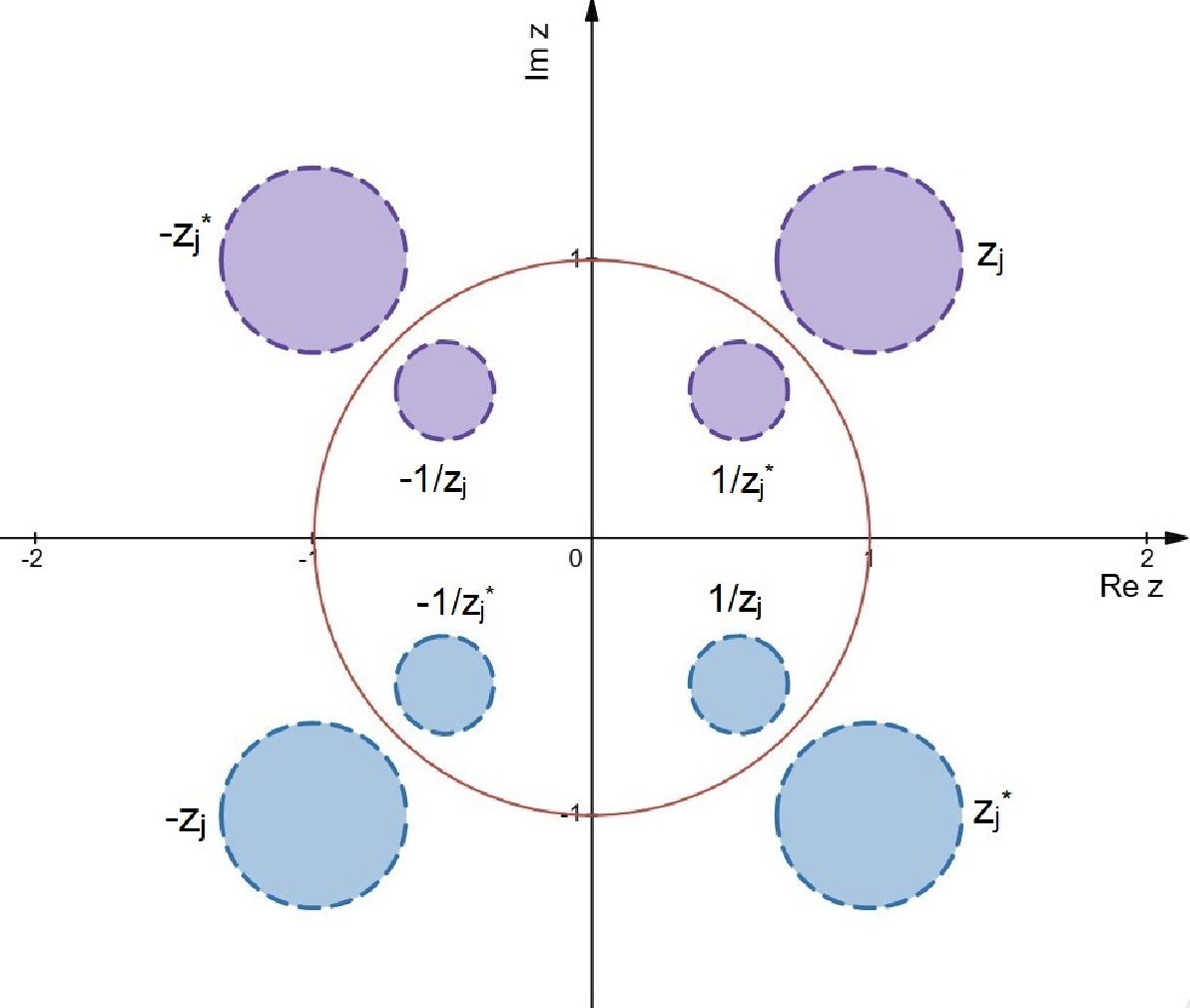}}}\hspace{-0.35in}
\vspace{0.05in}
\caption{Distribution of discrete spectrum $K\cup K^*$ and the parameters are $s_1=\frac34+i,s_2=\frac14,s_3=\frac13,m=1$.}
   \label{fig1}
\end{figure}

\begin{lemma}\label{le1}
For any open set $B_+$ containing the domain $\Omega_1$, the following identities hold:
\begin{eqnarray}
&& \no \lim\limits_{N\to\infty}\sum\limits_{j=1}^{N}\frac{c_je^{-2i\theta(z_j)}}{z-z_j}=\iint_{\Omega_1}\frac{r(\zeta,\zeta^*)e^{-2i\theta(\zeta)}}{2\pi i(z-\zeta)}d\zeta^*\wedge d\zeta, \vspace{0.2in} \\  \no \\
&&\no \lim\limits_{N\to\infty}\sum\limits_{j=1}^{N}\dfrac{c_j^*e^{-2i\theta(-z_j^*)}}{z+z_j^*}=\iint_{\Omega_1}\frac{r^*(\zeta,\zeta^*)e^{-2i\theta(-\zeta^*)}}{2\pi i(z+\zeta^*)}d\zeta^*\wedge d\zeta,\v\\ \no \\
&&\no \lim\limits_{N\to\infty}\sum\limits_{j=1}^{N}\dfrac{\frac{c_j^*}{z_j^{*2}}e^{-2i\theta(\frac{1}{z_j^*})}}{z-\frac{1}{z_j^*}}
=\iint_{\Omega_1}\dfrac{\frac{r^*(\zeta,\zeta^*)}{\zeta^{*2}}e^{-2i\theta(\frac{1}{\zeta^*})}}{2\pi i(z-\frac{1}{\zeta^*})}d\zeta^*\wedge d\zeta,\v\\ \no \\
&&\no \lim\limits_{N\to\infty}\sum\limits_{j=1}^{N}\dfrac{\frac{c_j}{z_j^{2}}e^{-2i\theta(-\frac{1}{z_j})}}{z+\frac{1}{z_j}}
=\iint_{\Omega_1}\dfrac{\frac{r(\zeta,\zeta^*)}{\zeta^{2}}e^{-2i\theta(-\frac{1}{\zeta})}}{2\pi i(z+\frac{1}{\zeta})}d\zeta^*\wedge d\zeta,\v\\ \no \\
&&\no \lim\limits_{N\to\infty}\sum\limits_{j=1}^{N}\frac{c_j^*e^{2i\theta(z_j^*)}}{z-z_j^*}=\iint_{\Omega_1}\frac{r^*(\zeta,\zeta^*)e^{2i\theta(\zeta^*)}}{2\pi i(z-\zeta^*)}d\zeta^*\wedge d\zeta,\v\\ \no \\
&&\no\lim\limits_{N\to\infty}\sum\limits_{j=1}^{N}\dfrac{c_je^{2i\theta(-z_j)}}{z+z_j}=\iint_{\Omega_1}\frac{r(\zeta,\zeta^*)e^{2i\theta(-\zeta)}}{2\pi i(z+\zeta)}d\zeta^*\wedge d\zeta,\v\\ \no \\
&&\no\lim\limits_{N\to\infty}\sum\limits_{j=1}^{N}\dfrac{\frac{c_j}{z_j^2}e^{2i\theta(\frac{1}{z_j})}}{z-\frac{1}{z_j}}=
\iint_{\Omega_1}\frac{\frac{r(\zeta,\zeta^*)}{{\zeta^2}}e^{2i\theta(\frac{1}{\zeta})}}{2\pi i(z-\frac{1}{\zeta})}d\zeta^*\wedge d\zeta,\v\\ \no \\
&& \lim\limits_{N\to\infty}\sum\limits_{j=1}^{N}\dfrac{\frac{c_j^*}{z_j^{*2}}e^{2i\theta(-\frac{1}{z_j^*})}}{z+\frac{1}{z_j^*}}=
\iint_{\Omega_1}\frac{\frac{r^*(\zeta,\zeta^*)}{\zeta^{*2}}e^{2i\theta(-\frac{1}{\zeta^*})}}{2\pi i(z+\frac{1}{\zeta^*})}d\zeta^*\wedge d\zeta,
\end{eqnarray}

uniformly for all $\mathbb{C}\setminus B_+$. The boundary $\partial\Omega_1$ is counterclockwise.

\end{lemma}

\begin{proof}
Using Eq.~(\ref{c-j}), we have
\begin{align}
\begin{array}{rl}
\d\lim\limits_{N\to\infty}\sum\limits_{j=1}^{N}\frac{c_je^{-2i\theta(z_j)}}{z-z_j}
=&\d \lim\limits_{N\to\infty}\sum\limits_{j=1}^{N}\frac{|\Omega_1|}{N}\frac{r(z_j,z_j^*)e^{-2i\theta(z_j)}}{\pi(z-z_j)} \v\\
=&\d \iint_{\Omega_1}\frac{r(\zeta,\zeta^*)e^{-2i\theta(\zeta)}}{2\pi i(z-\zeta)}d\zeta^*\wedge d\zeta.
\end{array}
\end{align}

Thus the proof is completed.
\end{proof}

\begin{lemma}\label{le2}
The following identities hold:
\begin{eqnarray}
&&\no \iint_{\Omega_1}\frac{r(\zeta,\zeta^*)e^{-2i\theta(\zeta)}}{2\pi i(z-\zeta)}d\zeta^*\wedge d\zeta=
\int_{\partial{\Omega_1}}\frac{(\zeta^*-s_1^*)^{n}r_1(\zeta)e^{-2i\theta(\zeta)}}{2\pi i(z-\zeta)}d\zeta,\v\\ \no \\
&&\no \iint_{\Omega_1}\frac{r^*(\zeta,\zeta^*)e^{-2i\theta(-\zeta^*)}}{2\pi i(z+\zeta^*)}d\zeta^*\wedge d\zeta
=-\int_{\partial{\Omega_1}}\frac{(\zeta-s_1)^{n}r_1^*(\zeta)e^{-2i\theta(-\zeta^*)}}{2\pi i(z+\zeta^*)}d\zeta^*,\v\\ \no \\
&&\no \iint_{\Omega_1}\dfrac{\frac{r^*(\zeta,\zeta^*)}{\zeta^{*2}}e^{-2i\theta(\frac{1}{\zeta^*})}}{2\pi i(z-\frac{1}{\zeta^*})}d\zeta^*\wedge d\zeta
=-\int_{\partial{\Omega_1}}\frac{(\zeta-s_1)^{n}r_1^*(\zeta)e^{-2i\theta(\frac{1}{\zeta^*})}}{2\pi i(z-\frac{1}{\zeta^*})\zeta^{*2}}d\zeta^*,\v\\ \no \\
&&\no \iint_{\Omega_1}\dfrac{\frac{r(\zeta,\zeta^*)}{\zeta^{2}}e^{-2i\theta(-\frac{1}{\zeta})}}{2\pi i(z+\frac{1}{\zeta})}d\zeta^*\wedge d\zeta
=\int_{\partial{\Omega_1}}\frac{(\zeta^*-s_1^*)^{n}r_1(\zeta)e^{-2i\theta(-\frac{1}{\zeta})}}{2\pi i(z+\frac{1}{\zeta})\zeta^2}d\zeta,\v\\ \no \\
&&\no \iint_{\Omega_1}\frac{r^*(\zeta,\zeta^*)e^{2i\theta(\zeta^*)}}{2\pi i(z-\zeta^*)}d\zeta^*\wedge d\zeta=-\int_{\partial{\Omega_1}}\frac{(\zeta-s_1)^{n}r_1^*(\zeta)e^{2i\theta(\zeta^*)}}{2\pi i(z-\zeta^*)}d\zeta^*,\v\\ \no \\
&&\no \iint_{\Omega_1}\frac{r(\zeta,\zeta^*)e^{2i\theta(-\zeta)}}{2\pi i(z+\zeta)}d\zeta^*\wedge d\zeta
=\int_{\partial{\Omega_1}}\frac{(\zeta^*-s_1^*)^{n}r_1(\zeta)e^{2i\theta(-\zeta)}}{2\pi i(z+\zeta)}d\zeta,\v\\ \no \\
&&\no \iint_{\Omega_1}\frac{\frac{r(\zeta,\zeta^*)}{{\zeta^2}}e^{2i\theta(\frac{1}{\zeta})}}{2\pi i(z-\frac{1}{\zeta})}d\zeta^*\wedge d\zeta
=\int_{\partial{\Omega_1}}\frac{(\zeta^*-s_1^*)^{n}r_1(\zeta)e^{2i\theta(\frac{1}{\zeta})}}{2\pi i(z-\frac{1}{\zeta})\zeta^2}d\zeta,\v\\  \no \\
&& \iint_{\Omega_1}\frac{\frac{r^*(\zeta,\zeta^*)}{\zeta^{*2}}e^{2i\theta(-\frac{1}{\zeta^*})}}{2\pi i(z+\frac{1}{\zeta^*})}d\zeta^*\wedge d\zeta
=-\int_{\partial{\Omega_1}}\frac{(\zeta-s_1)^{n}r_1^*(\zeta)e^{2i\theta(-\frac{1}{\zeta^*})}}{2\pi i(z+\frac{1}{\zeta^*})\zeta^{*2}}d\zeta^*,
\end{eqnarray}
uniformly for all $\mathbb{C}\setminus \Omega_1$. The boundary $\partial\Omega_1$ is counterclockwise.

\end{lemma}

\begin{proof}
Note that $r(z,z^*):=nz^{*(n-1)}r_1(z)$, using Green theorem, we have
\begin{align}\no
\begin{aligned}
\iint_{\Omega_1}\frac{r(\zeta,\zeta^*)e^{-2i\theta(\zeta)}}{2\pi i(z-\zeta)}d\zeta^*\wedge d\zeta
=&\iint_{\Omega_1}\frac{\overline{\partial}((\zeta^*-s_1^*)^{n})r_1(\zeta)e^{-2i\theta(\zeta)}}{2\pi i(z-\zeta)}d\zeta^*\wedge d\zeta\v\\
=&\iint_{\Omega_1}\overline{\partial}\left(\frac{(\zeta^*-s_1^*)^{n}r_1(\zeta)e^{-2i\theta(\zeta)}}{2\pi i(z-\zeta)}\right)d\zeta^*\wedge d\zeta\v\\
=&\int_{\partial{\Omega_1}}\frac{(\zeta^*-s_1^*)^{n}r_1(\zeta)e^{-2i\theta(\zeta)}}{2\pi i(z-\zeta)}d\zeta,
\end{aligned}
\end{align}
Thus the proof is completed.
\end{proof}

According to Lemmas \ref{le1} and \ref{le2}, we obtain a Riemann-Hilbert problem for $M^{(2)}(y,t;z):=\lim\limits_{N\to\infty}M^{(1)}$.

\begin{prop}\label{RH3}
Find a $2\times 2$ matrix function $M^{(2)}(y,t;z)$ that satisfies the following properties:

\begin{itemize}

 \item {} Analyticity: $M^{(2)}(y,t;z)$ is analytic in $\mathbb{C}\setminus(\Gamma_{1\pm}\cup\Gamma_{2\pm}\cup\Gamma_{3\pm}\cup\Gamma_{4\pm})$ and takes continuous boundary values on $\Gamma_{1\pm}\cup\Gamma_{2\pm}\cup\Gamma_{3\pm}\cup\Gamma_{4\pm}$.

 \item {} Jump condition: The boundary values on the jump contour $\Gamma_{1\pm}\cup\Gamma_{2\pm}\cup\Gamma_{3\pm}\cup\Gamma_{4\pm}$ are defined as
 \bee
M^{(2)}_{+}(y,t;z)=M^{(2)}_{-}(y,t;z)V_2(y,t;z),\quad z\in\Gamma_{1\pm}\cup\Gamma_{2\pm}\cup\Gamma_{3\pm}\cup\Gamma_{4\pm},
\ene
where
\bee\label{v2}
V_2(y,t;z)
=\begin{cases}
\left[\!\!\begin{array}{cc}
1& 0  \vspace{0.05in}\\
-\d\int_{\partial{\Omega_1}}\dfrac{(\zeta^*-s_1^*)^{n}r_1(\zeta)e^{-2i\theta(\zeta)}}{2\pi i(z-\zeta)}d\zeta& 1
\end{array}\!\!\right],\quad z\in\Gamma_{1+},\vspace{0.05in}\\
\left[\!\!\begin{array}{cc}
1& 0  \vspace{0.05in}\\
\d\int_{\partial{\Omega_1}}\dfrac{(\zeta-s_1)^{n}r_1^*(\zeta)e^{-2i\theta(-\zeta^*)}}{2\pi i(z+\zeta^*)}d\zeta^*& 1
\end{array}\!\!\right],\quad z\in\Gamma_{2+},\vspace{0.05in}\\
\left[\!\!\begin{array}{cc}
1& 0  \vspace{0.05in}\\
-\d\int_{\partial{\Omega_1}}\dfrac{(\zeta-s_1)^{n}r_1^*(\zeta)e^{-2i\theta(\frac{1}{\zeta^*})}}{2\pi i(z-\frac{1}{\zeta^*})\zeta^{*2}}d\zeta^*& 1
\end{array}\!\!\right],\quad z\in\Gamma_{3+},\vspace{0.05in}\\
\left[\!\!\begin{array}{cc}
1& 0  \vspace{0.05in}\\
\d\int_{\partial{\Omega_1}}\dfrac{(\zeta^*-s_1^*)^{n}r_1(\zeta)e^{-2i\theta(-\frac{1}{\zeta})}}{2\pi i(z+\frac{1}{\zeta})\zeta^2}d\zeta& 1
\end{array}\!\!\right],\quad z\in\Gamma_{4+},\vspace{0.05in}\\
\end{cases}
\ene
\bee\label{v2g}
V_2(y,t;z)
=\begin{cases}\left[\!\!\begin{array}{cc}
1& -\d\int_{\partial{\Omega_1}}\dfrac{(\zeta-s_1)^{n}r_1^*(\zeta)e^{2i\theta(\zeta^*)}}{2\pi i(z-\zeta^*)}d\zeta^*  \vspace{0.05in}\\
0& 1
\end{array}\!\!\right],\quad z\in\Gamma_{1-},\v\\
\left[\!\!\begin{array}{cc}
1&  \d\int_{\partial{\Omega_1}}\dfrac{(\zeta^*-s_1^*)^{n}r_1(\zeta)e^{2i\theta(-\zeta)}}{2\pi i(z+\zeta)}d\zeta \vspace{0.05in}\\
0& 1
\end{array}\!\!\right],\quad z\in\Gamma_{2-},\v\\
\left[\!\!\begin{array}{cc}
1& -\d\int_{\partial{\Omega_1}}\dfrac{(\zeta^*-s_1^*)^{n}r_1(\zeta)e^{2i\theta(\frac{1}{\zeta})}}{2\pi i(z-\frac{1}{\zeta})\zeta^2}d\zeta  \vspace{0.05in}\\
0& 1
\end{array}\!\!\right],\quad z\in\Gamma_{3-},\v\\
\left[\!\!\begin{array}{cc}
1& \d\int_{\partial{\Omega_1}}\dfrac{(\zeta-s_1)^{n}r_1^*(\zeta)e^{2i\theta(-\frac{1}{\zeta^*})}}{2\pi i(z+\frac{1}{\zeta^*})\zeta^{*2}}d\zeta^* \vspace{0.05in}\\
0& 1
\end{array}\!\!\right],\quad z\in\Gamma_{4-}.
\end{cases}
\ene

 \item {} Normalization:
\bee
 M^{(2)}=\left\{\begin{array}{ll}
    \mathbb{I}_2+O\left(1/z\right),  & z\to\infty, \v\\
    \dfrac{(q+1)^2}{m^2+(q+1)^2}\left(\begin{array}{cc} 1&  \frac{im}{q+1} \v\\
\frac{im}{q+1} & 1 \end{array}\right)\left[\mathbb{I}_2+\mu_1^{(0)}(z-i)\right]e^{\frac12 h_+\sigma_3}+\mathcal{O}((z-i)^2), & z\to i.
\end{array}\right.
\ene

\end{itemize}
\end{prop}

According to Eq.~(\ref{fanyan-1}), we can recover $u(x,t)$ by the following formula:
\bee
u(x,t)=\lim\limits_{z\rightarrow i}\frac{1}{z-i}\left(1-\dfrac{(M^{(2)}_{11}(z)+M^{(2)}_{21}(z))(M^{(2)}_{12}(z)+M^{(2)}_{22}(z))}
{(M^{(2)}_{11}(i)+M^{(2)}_{21}(i))(M^{(2)}_{12}(i)+M^{(2)}_{22}(i))}\right),
\ene
where
\bee\no
x(y,t)=y-\ln\left(\dfrac{M^{(2)}_{12}(i)+M^{(2)}_{22}(i)}{M^{(2)}_{11}(i)+M^{(2)}_{21}(i)}\right).
\ene

For the different quadrature domains, we will find different types of $N_\infty$-soliton asymptotic behaviors.

\subsection{$N_{\infty}$-soliton asymptotics: one-soliton solution}

\v {\it Case I}.---One-soliton solution. In this case, we choose $n=\ell=1$ such that we obtain the following proposition.

\begin{propo} Let $\zeta_0:=s_1+s_2$, then the solution of the RH problem \ref{RH3} can be used to generate the one-soliton solution $u_1(x,t)$ of the mCH equation (\ref{mch}) with the discrete spectrum $\zeta_0$ and norming constant $c_1=s_3^2r_1(\zeta_0)$.
\end{propo}

\begin{proof} At $\ell=n=1$, the boundary of $\Omega_1^*$ is described by
\bee\label{propo1-1}
z^*=s_1^*+\left(s_2^*+\dfrac{s_3^2}{z-\zeta_0}\right),\quad z\in\partial\Omega_1.
\ene

Substituting Eq.~(\ref{propo1-1}) into Eqs.~(\ref{v2}) and (\ref{v2g}), we obtain
\begin{eqnarray}
&&\no \int_{\partial{\Omega_1}}\frac{(\zeta^*-s_1^*)r_1(\zeta)e^{-2i\theta(\zeta)}}{2\pi i(z-\zeta)}d\zeta
=\frac{s_3^2r_1(\zeta_0)e^{-2i\theta(\zeta_0)}}{z-\zeta_0},\v\\ \no \\
&&\no \int_{\partial{\Omega_1}}\frac{(\zeta-s_1)r_1^*(\zeta)e^{-2i\theta(-\zeta^*)}}{2\pi i(z+\zeta^*)}d\zeta^*
=-\frac{s_3^2r_1^*(\zeta_0)e^{-2i\theta(-\zeta_0^*)}}{z+\zeta_0^*},\v\\  \no \\
&&\no \int_{\partial{\Omega_1}}\frac{(\zeta-s_1)r_1^*(\zeta)e^{-2i\theta(\frac{1}{\zeta^*})}}{2\pi i(z-\frac{1}{\zeta^*})\zeta^{*2}}d\zeta^*
=-\frac{s_3^2r_1^*(\zeta_0)e^{-2i\theta(\frac{1}{\zeta_0^*})}}{(z-\frac{1}{\zeta_0^*})\zeta_0^{*2}},\v\\  \no \\
&&\no \int_{\partial{\Omega_1}}\frac{(\zeta^*-s_1^*)r_1(\zeta)e^{-2i\theta(-\frac{1}{\zeta})}}{2\pi i(z+\frac{1}{\zeta})\zeta^2}d\zeta
=\frac{s_3^2r_1(\zeta_0)e^{-2i\theta(-\frac{1}{\zeta_0})}}{(z+\frac{1}{\zeta_0})\zeta_0^2},\v\\ \no \\
&&\no \int_{\partial{\Omega_1}}\frac{(\zeta-s_1)r_1^*(\zeta)e^{2i\theta(\zeta^*)}}{2\pi i(z-\zeta^*)}d\zeta^*
=-\frac{s_3^2r_1^*(\zeta_0)e^{2i\theta(\zeta_0^*)}}{z-\zeta_0^*},\v\\ \no \\
&&\no \int_{\partial{\Omega_1}}\frac{(\zeta^*-s_1^*)r_1(\zeta)e^{2i\theta(-\zeta)}}{2\pi i(z+\zeta)}d\zeta
=\frac{s_3^2r_1(\zeta)e^{2i\theta(-\zeta_0)}}{z+\zeta_0},\v\\ \no \\
&&\no \int_{\partial{\Omega_1}}\frac{(\zeta^*-s_1^*)r_1(\zeta)e^{2i\theta(\frac{1}{\zeta})}}{2\pi i(z-\frac{1}{\zeta})\zeta^2}d\zeta
=\frac{s_3^2r_1(\zeta_0)e^{2i\theta(\frac{1}{\zeta_0})}}{(z-\frac{1}{\zeta_0})\zeta_0^2},\v\\ \no \\
&& \int_{\partial{\Omega_1}}\frac{(\zeta-s_1)r_1^*(\zeta)e^{2i\theta(-\frac{1}{\zeta^*})}}{2\pi i(z+\frac{1}{\zeta^*})\zeta^{*2}}d\zeta^*
=-\frac{s_3^2r_1^*(\zeta_0)e^{2i\theta(-\frac{1}{\zeta_0^*})}}{(z+\frac{1}{\zeta_0^*})\zeta_0^{*2}}.
\label{propo1-2}
\end{eqnarray}
Then the jump matrixes given by Eqs.~(\ref{v2}) and (\ref{v2g}) can be rewritten as:
\bee \label{v22}
V_2(y,t;z)|_{\ell=n=1}
=\begin{cases}
\left[\!\!\begin{array}{cc}
1& 0  \vspace{0.05in}\\
-\dfrac{s_3^2r_1(\zeta_0)e^{-2i\theta(\zeta_0)}}{z-\zeta_0}& 1
\end{array}\!\!\right],\quad z\in\Gamma_{1+},\vspace{0.05in}\\
\left[\!\!\begin{array}{cc}
1& 0  \vspace{0.05in}\\
-\dfrac{s_3^2r_1^*(\zeta_0)e^{-2i\theta(-\zeta_0^*)}}{z+\zeta_0^*}& 1
\end{array}\!\!\right],\quad z\in\Gamma_{2+},\vspace{0.05in}\\
\left[\!\!\begin{array}{cc}
1& 0  \vspace{0.05in}\\
\dfrac{s_3^2r_1^*(\zeta_0)e^{-2i\theta(\frac{1}{\zeta_0^*})}}{(z-\frac{1}{\zeta_0^*})\zeta_0^{*2}}& 1
\end{array}\!\!\right],\quad z\in\Gamma_{3+},\vspace{0.05in}\\
\left[\!\!\begin{array}{cc}
1& 0  \vspace{0.05in}\\
\dfrac{s_3^2r_1(\zeta_0)e^{-2i\theta(-\frac{1}{\zeta_0})}}{(z+\frac{1}{\zeta_0})\zeta_0^2}& 1
\end{array}\!\!\right],\quad z\in\Gamma_{4+},\vspace{0.05in}\\
\end{cases}
\ene
\bee\label{v22g}
V_2(y,t;z)|_{\ell=n=1}
=\begin{cases}
\left[\!\!\begin{array}{cc}
1& \dfrac{s_3^2r_1^*(\zeta_0)e^{2i\theta(\zeta_0^*)}}{z-\zeta_0^*}  \vspace{0.05in}\\
0& 1
\end{array}\!\!\right],\quad z\in\Gamma_{1-},\v\\
\left[\!\!\begin{array}{cc}
1& \dfrac{s_3^2r_1(\zeta)e^{2i\theta(-\zeta_0)}}{z+\zeta_0} \vspace{0.05in}\\
0& 1
\end{array}\!\!\right],\quad z\in\Gamma_{2-},\v\\
\left[\!\!\begin{array}{cc}
1& -\dfrac{s_3^2r_1(\zeta_0)e^{2i\theta(\frac{1}{\zeta_0})}}{(z-\frac{1}{\zeta_0})\zeta_0^2}  \vspace{0.05in}\\
0& 1
\end{array}\!\!\right],\quad z\in\Gamma_{3-},\v\\
\left[\!\!\begin{array}{cc}
1& -\dfrac{s_3^2r_1^*(\zeta_0)e^{2i\theta(-\frac{1}{\zeta_0^*})}}{(z+\frac{1}{\zeta_0^*})\zeta_0^{*2}} \vspace{0.05in}\\
0& 1
\end{array}\!\!\right],\quad z\in\Gamma_{4-}.
\end{cases}
\ene

Let the discrete spectrum be $z_1:=\zeta_0$ and the norming constant $c_1=s_3^2r_1(\zeta_0)$. Then the solution of Riemann-Hilbert problem \ref{RH3} with the jump matrix given by Eqs.~(\ref{v22}) and (\ref{v22g}) can generate the one-soliton solution $u_1(x,t)$ of the mCH equation (\ref{mch}).
\end{proof}

\subsection{$N_{\infty}$-soliton asymptotics: the n-soliton solution}

\v {\it  Case II}.---$n$-soliton solution. In this case, if we choose $n=\ell$. Then we obtain the following proposition.

\begin{propo} Let $\{\zeta_1,\zeta_2,\cdots,\zeta_n\}$ being the solutions to the algebraic equation $(z-s_1)^n=s_2$, then the solution of the RH problem \ref{RH3} can lead to the $n$-soliton solution $u_n(x,t)$ of the mCH equation (\ref{mch}) with discrete spectra $\zeta_j,\, j=1,\cdots,n$ and norming constants $c_j=\dfrac{s_3^2r_1(\zeta_j)}{\prod_{k\neq j}(\zeta_j-\zeta_k)},j=1,\cdots,n$.
\end{propo}

\begin{proof} At $\ell=n$, the boundary of $\Omega_1^*$ is described by
\bee\label{propo2-1}
z^*=s_1^*+\left(s_2^*+\dfrac{s_3^2}{(z-s_1)^m-s_2}\right)^{\frac{1}{n}},\quad z\in\partial\Omega_1.
\ene

Substituting Eq.~(\ref{propo2-1}) into Eqs.~(\ref{v2}) and (\ref{v2g}), we obtain
\begin{eqnarray}
&& \no \int_{\partial{\Omega_1}}\frac{(\zeta^*-s_1^*)^nr_1(\zeta)e^{-2i\theta(\zeta)}}{2\pi i(z-\zeta)}d\zeta
=\sum_{j=1}^n\frac{s_3^2r_1(\zeta_j)e^{-2i\theta(\zeta_j)}}{(z-\zeta_j)\prod_{k\neq j}(\zeta_j-\zeta_k)},\v\\ \no \\
&& \no  \int_{\partial{\Omega_1}}\frac{(\zeta-s_1)^nr_1^*(\zeta)e^{-2i\theta(-\zeta^*)}}{2\pi i(z+\zeta^*)}d\zeta^*
=-\sum_{j=1}^n\frac{s_3^2r_1^*(\zeta_j)e^{-2i\theta(-\zeta_j^*)}}{(z+\zeta_j^*)\prod_{k\neq j}(\zeta_j^*-\zeta_k^*)},\v\\ \no \\
&& \no  \int_{\partial{\Omega_1}}\frac{(\zeta-s_1)^nr_1^*(\zeta)e^{-2i\theta(\frac{1}{\zeta^*})}}{2\pi i(z-\frac{1}{\zeta^*})\zeta^{*2}}d\zeta^*
=-\sum_{j=1}^n\frac{s_3^2r_1^*(\zeta_j)e^{-2i\theta(\frac{1}{\zeta_j^*})}}{(z-\frac{1}{\zeta_j^*})\zeta_j^{*2}\prod_{k\neq j}(\zeta_j^*-\zeta_k^*)},\v\\  \no \\
&& \no  \int_{\partial{\Omega_1}}\frac{(\zeta^*-s_1^*)^nr_1(\zeta)e^{-2i\theta(-\frac{1}{\zeta})}}{2\pi i(z+\frac{1}{\zeta})\zeta^2}d\zeta
=\sum_{j=1}^n\frac{s_3^2r_1(\zeta_j)e^{-2i\theta(-\frac{1}{\zeta_j})}}{(z+\frac{1}{\zeta_j})\zeta_j^2\prod_{k\neq j}(\zeta_j-\zeta_k)},\v\\ \no \\
&& \no  \int_{\partial{\Omega_1}}\frac{(\zeta-s_1)^nr_1^*(\zeta)e^{2i\theta(\zeta^*)}}{2\pi i(z-\zeta^*)}d\zeta^*
=-\sum_{j=1}^n\frac{s_3^2r_1^*(\zeta_j)e^{2i\theta(\zeta_j^*)}}{(z-\zeta_j^*)\prod_{k\neq j}(\zeta_j^*-\zeta_k^*)},\v\\ \no \\
&& \no  \int_{\partial{\Omega_1}}\frac{(\zeta^*-s_1^*)^nr_1(\zeta)e^{2i\theta(-\zeta)}}{2\pi i(z+\zeta)}d\zeta
=\sum_{j=1}^n\frac{s_3^2r_1(\zeta)e^{2i\theta(-\zeta_j)}}{(z+\zeta_j)\prod_{k\neq j}(\zeta_j-\zeta_k)},\v\\ \no \\
&& \no  \int_{\partial{\Omega_1}}\frac{(\zeta^*-s_1^*)^nr_1(\zeta)e^{2i\theta(\frac{1}{\zeta})}}{2\pi i(z-\frac{1}{\zeta})\zeta^2}d\zeta
=\sum_{j=1}^n\frac{s_3^2r_1(\zeta_j)e^{2i\theta(\frac{1}{\zeta_j})}}{(z-\frac{1}{\zeta_j})\zeta_j^2\prod_{k\neq j}(\zeta_j-\zeta_k)},\v\\ \no \\
&& \int_{\partial{\Omega_1}}\frac{(\zeta-s_1)^nr_1^*(\zeta)e^{2i\theta(-\frac{1}{\zeta^*})}}{2\pi i(z+\frac{1}{\zeta^*})\zeta^{*2}}d\zeta^*
=-\sum_{j=1}^n\frac{s_3^2r_1^*(\zeta_j)e^{2i\theta(-\frac{1}{\zeta_j^*})}}{(z+\frac{1}{\zeta_j^*})\zeta_j^{*2}\prod_{k\neq j}(\zeta_j^*-\zeta_k^*)}.
\label{propo2-2}
\end{eqnarray}
Then the jump matrixes given by Eqs.~(\ref{v2}) and (\ref{v2g}) can be rewritten as:
\bee \label{v23}
V_2(y,t;z)|_{\ell=n}
=\begin{cases}
\left[\!\!\begin{array}{cc}
1& 0  \vspace{0.05in}\\
-\d\sum_{j=1}^n\frac{s_3^2r_1(\zeta_j)e^{-2i\theta(\zeta_j)}}{(z-\zeta_j)\prod_{k\neq j}(\zeta_j-\zeta_k)}& 1
\end{array}\!\!\right],\quad z\in\Gamma_{1+},\vspace{0.05in}\\
\left[\!\!\begin{array}{cc}
1& 0  \vspace{0.05in}\\
-\d \sum_{j=1}^n\frac{s_3^2r_1^*(\zeta_j)e^{-2i\theta(-\zeta_j^*)}}{(z+\zeta_j^*)\prod_{k\neq j}(\zeta_j^*-\zeta_k^*)}& 1
\end{array}\!\!\right],\quad z\in\Gamma_{2+},\vspace{0.05in}\\
\end{cases}
\ene
\bee \label{v231}
V_2(y,t;z)|_{\ell=n}
=\begin{cases}
\left[\!\!\begin{array}{cc}
1& 0  \vspace{0.05in}\\
\d \sum_{j=1}^n\frac{s_3^2r_1^*(\zeta_j)e^{-2i\theta(\frac{1}{\zeta_j^*})}}{(z-\frac{1}{\zeta_j^*})\zeta_j^{*2}\prod_{k\neq j}(\zeta_j^*-\zeta_k^*)}& 1
\end{array}\!\!\right],\quad z\in\Gamma_{3+},\vspace{0.05in}\\
\left[\!\!\begin{array}{cc}
1& 0  \vspace{0.05in}\\
\d \sum_{j=1}^n\frac{s_3^2r_1(\zeta_j)e^{-2i\theta(-\frac{1}{\zeta_j})}}{(z+\frac{1}{\zeta_j})\zeta_j^2\prod_{k\neq j}(\zeta_j-\zeta_k)}& 1
\end{array}\!\!\right],\quad z\in\Gamma_{4+},\vspace{0.05in}\\
\end{cases}
\ene
\bee \label{v23g}
V_2(y,t;z)|_{\ell=n}
=\begin{cases}
\left[\!\!\begin{array}{cc}
1& \d \sum_{j=1}^n\frac{s_3^2r_1^*(\zeta_j)e^{2i\theta(\zeta_j^*)}}{(z-\zeta_j^*)\prod_{k\neq j}(\zeta_j^*-\zeta_k^*)}  \vspace{0.05in}\\
0& 1
\end{array}\!\!\right],\quad z\in\Gamma_{1-},\v\\
\left[\!\!\begin{array}{cc}
1& \d \sum_{j=1}^n\frac{s_3^2r_1(\zeta)e^{2i\theta(-\zeta_j)}}{(z+\zeta_j)\prod_{k\neq j}(\zeta_j-\zeta_k)} \vspace{0.05in}\\
0& 1
\end{array}\!\!\right],\quad z\in\Gamma_{2-},
\end{cases}
\ene
\bee
V_2(y,t;z)|_{\ell=n}
=\begin{cases}
\left[\!\!\begin{array}{cc}
1& -\d \sum_{j=1}^n\frac{s_3^2r_1(\zeta_j)e^{2i\theta(\frac{1}{\zeta_j})}}{(z-\frac{1}{\zeta_j})\zeta_j^2\prod_{k\neq j}(\zeta_j-\zeta_k)}  \vspace{0.05in}\\
0& 1
\end{array}\!\!\right],\quad z\in\Gamma_{3-},\v\\
\left[\!\!\begin{array}{cc}
1& -\d \sum_{j=1}^n\frac{s_3^2r_1^*(\zeta_j)e^{2i\theta(-\frac{1}{\zeta_j^*})}}{(z+\frac{1}{\zeta_j^*})\zeta_j^{*2}\prod_{k\neq j}(\zeta_j^*-\zeta_k^*)} \vspace{0.05in}\\
0& 1
\end{array}\!\!\right],\quad z\in\Gamma_{4-}.
\end{cases}
\ene

Let the discrete spectra be $z_j:=\zeta_j,j=1,\cdots,n$ and the norming constants $c_j=\dfrac{s_3^2r_1(\zeta_j)}{\prod_{k\neq j}(\zeta_j-\zeta_k)},j=1,\cdots,n$, then the solution of Riemann-Hilbert problem \ref{RH3} with the jump matrixes given by
 Eqs.~(\ref{v23}) and (\ref{v23g}) can generate the $n$-soliton solution $u_n(x,t)$ of the mCH equation (\ref{mch}) with discrete spectrums $z_j,j=1,\cdots,n$ and norming constants $c_j,j=1,\cdots,n$.
\end{proof}

\section{$N_{\infty}$-soliton asymptotics: the line domain}

In this section, we care about the $N_{\infty}$ asymptotic situation of $N$-soliton solution, under the additional assumptions:

\begin{itemize}

 \item {} The poles $\{z_j\}_{j=1}^N$ are sampled from a smooth density function $\rho(z)$ so that $\int_{a}^{-iz_j}\rho(\zeta)d\zeta=\frac{j}{N},j=1,\cdots,N$.

 \item {} The coefficients $\{c_j\}_{j=1}^N$ satisfy
\bee\label{c-j-2}
c_j=\dfrac{(b-a)r(z_j)}{N\pi},\quad b>a>1,
\ene
where $r(z)$ is a real-valued, continuous and non-vanishing function for $z\in(ia,ib)$, with the symmetry $r(z^*)=r(z)=r(-\frac{1}{z})=r(-\frac{1}{z^*})$.

\end{itemize}

\begin{lemma}\label{le3}
For any open set $A_+(B_+)$ containing the interval $[ia,ib]([\frac{i}{b},\frac{i}{a}])$, and any open set $A_-(B_-)$ containing the interval $[-ib,-ia]([-\frac{i}{a},-\frac{i}{b}])$, the following identities hold:
\bee
\lim\limits_{N\to\infty}\sum\limits_{j=1}^{N}\frac{c_je^{-2i\theta(z_j)}}{z-z_j}
=-\int_{ia}^{ib}\dfrac{ir(w)e^{-2i\theta(w)}}{\pi(z-w)}dw,
\ene
uniformly for all $\mathbb{C}\setminus A_+$.
\bee
\lim\limits_{N\to\infty}\sum\limits_{j=1}^{N}\dfrac{\frac{c_j}{z_j^{2}}e^{-2i\theta(-\frac{1}{z_j})}}{z+\frac{1}{z_j}}
=\int_{\frac{i}{b}}^{\frac{i}{a}}\dfrac{ir(w)e^{-2i\theta(w)}}{\pi(z-w)}dw,
\ene
uniformly for all $\mathbb{C}\setminus B_+$.
\bee
\lim\limits_{N\to\infty}\sum\limits_{j=1}^{N}\frac{c_j^*e^{2i\theta(z_j^*)}}{z-z_j^*}=
-\int_{-ib}^{-ia}\dfrac{ir(w)e^{2i\theta(w)}}{\pi(z-w)}dw,
\ene
uniformly for all $\mathbb{C}\setminus A_-$.
\bee
\lim\limits_{N\to\infty}\sum\limits_{j=1}^{N}\dfrac{\frac{c_j}{z_j^2}e^{2i\theta(\frac{1}{z_j})}}{z-\frac{1}{z_j}}
=\int_{-\frac{i}{a}}^{-\frac{i}{b}}\dfrac{ir(w)e^{2i\theta(w)}}{\pi(z-w)}dw,
\ene
uniformly for all $\mathbb{C}\setminus B_-$. The open intervals $(ia,ib),(-\frac{i}{a},-\frac{i}{b}),(\frac{i}{b},\frac{i}{a})$ and $(-ib,-ia)$ are both oriented upwards.

\end{lemma}

\begin{proof}
Using Eq.~(\ref{c-j-2}), we have
\begin{align}\no
\begin{aligned}
\lim\limits_{N\to\infty}\sum\limits_{j=1}^{N}\frac{c_je^{-2i\theta(z_j)}}{z-z_j}
=&\lim\limits_{N\to\infty}\sum\limits_{j=1}^{N}\frac{i(b-a)}{N}\frac{-ir(z_j)e^{-2i\theta(z_j)}}{\pi(z-z_j)}\v\\
=&-\int_{ia}^{ib}\dfrac{ir(w)e^{-2i\theta(w)}}{\pi(z-w)}dw,
\end{aligned}
\end{align}
and
\begin{align}\no
\begin{aligned}
\lim\limits_{N\to\infty}\sum\limits_{j=1}^{N}\dfrac{\frac{c_j}{z_j^{2}}e^{-2i\theta(-\frac{1}{z_j})}}{z+\frac{1}{z_j}}
=&\lim\limits_{N\to\infty}\sum\limits_{j=1}^{N}\frac{i(b-a)}{N}\dfrac{-ir(z_j)e^{-2i\theta(-\frac{1}{z_j})}}{z_j^2\pi(z+\frac{1}{z_j})}\v\\
=&-\int_{ia}^{ib}\dfrac{ir(\lambda)e^{-2i\theta(-\frac{1}{\lambda})}}{\lambda^2\pi(z+\frac{1}{\lambda})}d\lambda\v\\
=&\int_{\frac{i}{b}}^{\frac{i}{a}}\dfrac{ir(w)e^{-2i\theta(w)}}{\pi(z-w)}dw.
\end{aligned}
\end{align}

Thus the proof is completed.
\end{proof}

We define a closed curve $\gamma_{1+}(\gamma_{2+})$ with a very small radius encircling the poles $\{z_j\}_{j=1}^N(-\frac{1}{z_j})$ counterclockwise in the upper half plane $\mathbb{C}^+$, and a closed curve $\Gamma_{1-}(\Gamma_{2-})$ with a very small radius encircling the poles $\{z_j^*\}_{j=1}^N(\frac{1}{z_j})$ counterclockwise in the lower half plane $\mathbb{C}^-$. According to Lemma \ref{le3}, the jump matrix $V_1(y,t;z)$ defined by Eq.~(\ref{V1-1}) can be rewrite as:
\bee
V_1(y,t;z)|_{N\to\infty}
=\begin{cases}
\left[\!\!\begin{array}{cc}
1& 0  \vspace{0.05in}\\
\displaystyle\int_{ia}^{ib}\dfrac{ir(w)e^{-2i\theta(w)}}{\pi(z-w)}dw& 1
\end{array}\!\!\right],\quad z\in\gamma_{1+},\vspace{0.05in}\\
\left[\!\!\begin{array}{cc}
1& 0  \vspace{0.05in}\\
\displaystyle\int_{\frac{i}{b}}^{\frac{i}{a}}\dfrac{ir(w)e^{-2i\theta(w)}}{\pi(z-w)}dw& 1
\end{array}\!\!\right],\quad z\in\gamma_{2+},\vspace{0.05in}\\
\left[\!\!\begin{array}{cc}
1& -\displaystyle\int_{-ib}^{-ia}\dfrac{ir(w)e^{2i\theta(w)}}{\pi(z-w)}dw  \vspace{0.05in}\\
0& 1
\end{array}\!\!\right],\quad z\in\gamma_{1-},\v\\
\left[\!\!\begin{array}{cc}
1&-\displaystyle\int_{-\frac{i}{a}}^{-\frac{i}{b}}\dfrac{ir(w)e^{2i\theta(w)}}{\pi(z-w)}dw  \vspace{0.05in}\\
0& 1
\end{array}\!\!\right],\quad z\in\gamma_{2-}.
\end{cases}
\ene

Make the following transformation:
\bee
M^{(3)}(y,t;z)
=\begin{cases}
M^{(1)}(y,t;z)\left[\!\!\begin{array}{cc}
1& 0  \vspace{0.05in}\\
-\displaystyle\int_{ia}^{ib}\dfrac{ir(w)e^{-2i\theta(w)}}{\pi(z-w)}dw& 1
\end{array}\!\!\right],\quad z~\mathrm{within}~\gamma_{1+},\vspace{0.05in}\\
M^{(1)}(y,t;z)\left[\!\!\begin{array}{cc}
1& 0  \vspace{0.05in}\\
-\displaystyle\int_{\frac{i}{b}}^{\frac{i}{a}}\dfrac{ir(w)e^{-2i\theta(w)}}{\pi(z-w)}dw& 1
\end{array}\!\!\right],\quad z~\mathrm{within}~\gamma_{2+},\vspace{0.05in}\\
\end{cases}
\ene
\bee
M^{(3)}(y,t;z)
=\begin{cases}
M^{(1)}(y,t;z)\left[\!\!\begin{array}{cc}
1& \displaystyle\int_{-ib}^{-ia}\dfrac{ir(w)e^{2i\theta(w)}}{\pi(z-w)}dw  \vspace{0.05in}\\
0& 1
\end{array}\!\!\right],\quad z~\mathrm{within}~\gamma_{1-},\v\\
M^{(1)}(y,t;z)\left[\!\!\begin{array}{cc}
1&\displaystyle\int_{-\frac{i}{a}}^{-\frac{i}{b}}\dfrac{ir(w)e^{2i\theta(w)}}{\pi(z-w)}dw  \vspace{0.05in}\\
0& 1
\end{array}\!\!\right],\quad z~\mathrm{within}~\gamma_{2-},\v\\
M^{(1)}(y,t;z),\quad \mathrm{otherwise}.
\end{cases}
\ene

Then matrix function $M^{(3)}(x,t;z)$ satisfies the following Riemann-Hilbert problem.

\begin{prop}\label{RH4}
Find a $2\times 2$ matrix function $M^{(3)}(y,t;z)$ that satisfies the following properties:

\begin{itemize}

 \item {} Analyticity: $M^{(3)}(y,t;z)$ is analytic in $\mathbb{C}\setminus((ia,ib)\cup(\frac{i}{b},\frac{i}{a})\cup(-ib,-ia)\cup(-\frac{i}{a},-\frac{i}{b}))$ and takes continuous boundary values on $(ia,ib)\cup(\frac{i}{b},\frac{i}{a})\cup(-ib,-ia)\cup(-\frac{i}{a},-\frac{i}{b})$(The directions of these open intervals are all facing upwards).

 \item {} Jump condition: The boundary values on the jump contour $(ia,ib)\cup(\frac{i}{b},\frac{i}{a})\cup(-ib,-ia)\cup(-\frac{i}{a},-\frac{i}{b})$ are defined as
 \bee
M^{(3)}_{+}(y,t;z)=M^{(3)}_{-}(y,t;z)V_3(y,t;z),\quad z\in(ia,ib)\cup(\frac{i}{b},\frac{i}{a})\cup(-ib,-ia)\cup(-\frac{i}{a},-\frac{i}{b}),
\ene
where
\bee
V_3(y,t;z)
=\begin{cases}
\left[\!\!\begin{array}{cc}
1& 0  \vspace{0.05in}\\
-2r(z)e^{-2i\theta(z)}& 1
\end{array}\!\!\right],\quad z\in(ia,ib),\vspace{0.05in}\\
\left[\!\!\begin{array}{cc}
1& 0  \vspace{0.05in}\\
-2r(z)e^{-2i\theta(z)}& 1
\end{array}\!\!\right],\quad z\in(\frac{i}{b},\frac{i}{a}),\vspace{0.05in}\\
\left[\!\!\begin{array}{cc}
1& 2r(z)e^{2i\theta(z)}  \vspace{0.05in}\\
0& 1
\end{array}\!\!\right],\quad z\in(-ib,-ia),\v\\
\left[\!\!\begin{array}{cc}
1&  2r(z)e^{2i\theta(z)}  \vspace{0.05in}\\
0& 1
\end{array}\!\!\right],\quad z\in(-\frac{i}{a},-\frac{i}{b}).
\end{cases}
\ene

 \item {} Normalization:
\bee
 M^{(3)}=\left\{\begin{array}{ll}
    \mathbb{I}_2+O\left(1/z\right),  & z\to\infty, \v\\
    \dfrac{(q+1)^2}{m^2+(q+1)^2}\left(\begin{array}{cc} 1&  \frac{im}{q+1} \v\\
\frac{im}{q+1} & 1 \end{array}\right)\left(\mathbb{I}_2+\mu_1^{(0)}(z-i)\right)e^{\frac12 h_+\sigma_3}+\mathcal{O}((z-i)^2), & z\to i.
\end{array}\right.
\ene
\end{itemize}
\end{prop}

According to Eq.~(\ref{fanyan-1}), we can recover $u(x,t)$ by the following formula:
\bee\label{fanyan-2}
u(x,t)=\lim\limits_{z\rightarrow i}\frac{1}{z-i}\left(1-\dfrac{(M^{(3)}_{11}(z)+M^{(3)}_{21}(z))(M^{(3)}_{12}(z)+M^{(3)}_{22}(z))}{(M^{(3)}_{11}(i)+M^{(3)}_{21}(i))(M^{(3)}_{12}(i)+M^{(3)}_{22}(i))}\right).
\ene
where
\bee\no
x(y,t)=y-\ln\left(\dfrac{M^{(3)}_{12}(i)+M^{(3)}_{22}(i)}{M^{(3)}_{11}(i)+M^{(3)}_{21}(i)}\right).
\ene

Below, we will study the asymptotic behavior of  $u(x):=u(x,0)$. Firstly, we construct the following Riemann-Hilbert problem:

\begin{prop}\label{RH5}
Find a $2\times 2$ matrix function $N(y;z)$ that satisfies the following properties:

\begin{itemize}

 \item {} Analyticity: $N(y;z)$ is analytic in $\mathbb{C}\setminus((ia,ib)\cup(\frac{i}{b},\frac{i}{a})\cup(-ib,-ia)\cup(-\frac{i}{a},-\frac{i}{b}))$ and takes continuous boundary values on $(ia,ib)\cup(\frac{i}{b},\frac{i}{a})\cup(-ib,-ia)\cup(-\frac{i}{a},-\frac{i}{b})$(The directions of these open intervals are all facing upwards).

 \item {} Jump condition: The boundary values on the jump contour $(ia,ib)\cup(\frac{i}{b},\frac{i}{a})\cup(-ib,-ia)\cup(-\frac{i}{a},-\frac{i}{b})$ are defined as
 \bee
N_{+}(y;z)=N_{-}(y;z)J(y;z),\quad z\in(ia,ib)\cup(\frac{i}{b},\frac{i}{a})\cup(-ib,-ia)\cup(-\frac{i}{a},-\frac{i}{b}),
\ene
where
\bee
J(y;z)
=\begin{cases}
\left[\!\!\begin{array}{cc}
1& 0  \vspace{0.05in}\\
-2r(z)e^{\frac{i}{2}(z-\frac{1}{z})y}& 1
\end{array}\!\!\right],\quad z\in(ia,ib),\vspace{0.05in}\\
\left[\!\!\begin{array}{cc}
1& 0  \vspace{0.05in}\\
-2r(z)e^{\frac{i}{2}(z-\frac{1}{z})y}& 1
\end{array}\!\!\right],\quad z\in(\frac{i}{b},\frac{i}{a}),\vspace{0.05in}\\
\left[\!\!\begin{array}{cc}
1& 2r(z)e^{-\frac{i}{2}(z-\frac{1}{z})y}  \vspace{0.05in}\\
0& 1
\end{array}\!\!\right],\quad z\in(-ib,-ia),\v\\
\left[\!\!\begin{array}{cc}
1&  2r(z)e^{-\frac{i}{2}(z-\frac{1}{z})y}  \vspace{0.05in}\\
0& 1
\end{array}\!\!\right],\quad z\in(-\frac{i}{a},-\frac{i}{b}).
\end{cases}
\ene

 \item {} Normalization:
\bee
N=\left\{\begin{array}{ll}
    \mathbb{I}_2+O\left(1/z\right),  & z\to\infty, \v\\
    \dfrac{(q+1)^2}{m^2+(q+1)^2}\left(\begin{array}{cc} 1&  \frac{im}{q+1} \v\\
\frac{im}{q+1} & 1 \end{array}\right)\left(\mathbb{I}_2+\mu_1^{(0)}(z-i)\right)e^{\frac12 h_+\sigma_3}+\mathcal{O}((z-i)^2), & z\to i.
\end{array}\right.
\ene

\end{itemize}
\end{prop}

According to Eq.~(\ref{fanyan-2}), we can recover $u(x)$ by the following formula:
\bee\label{fanyan-3}
u(x)=\lim\limits_{z\rightarrow i}\frac{1}{z-i}\left(1-\dfrac{(N_{11}(z)+N_{21}(z))(N_{12}(z)+N_{22}(z))}{(N_{11}(i)+N_{21}(i))(N_{12}(i)+N_{22}(i))}\right).
\ene
where
\bee\no
x(y)=y-\ln\left(\dfrac{N_{12}(i)+N_{22}(i)}{N_{11}(i)+N_{21}(i)}\right).
\ene

Make the following transformation:
\bee
N^{(1)}(z)=N(iz),\quad r_2(z)=2r(iz).
\ene

\begin{figure}[!t]
    \centering
 \vspace{-0.15in}
  {\scalebox{0.55}[0.55]{\includegraphics{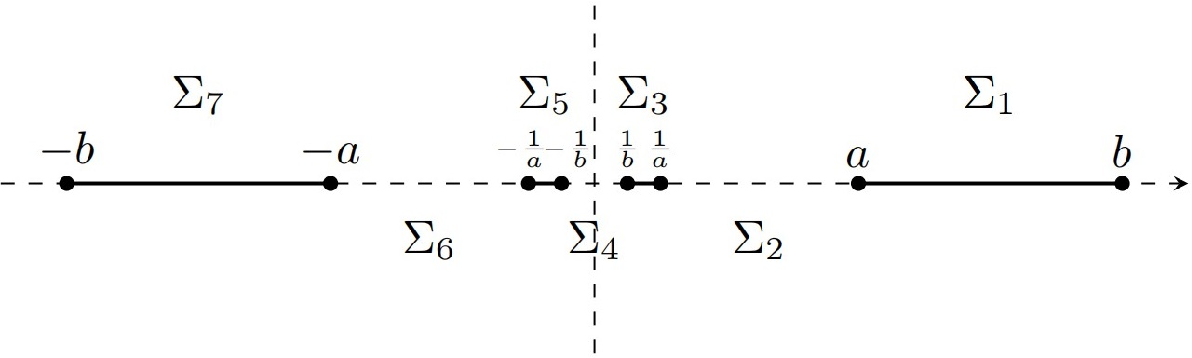}}}\hspace{-0.35in}
\vspace{0.05in}
\caption{Distribution of discrete spectrum for the line region.}
   \label{fig3}
\end{figure}

Then matrix function $N^{(1)}(z)$ satisfies the following Riemann-Hilbert problem.

\begin{prop}\label{RH6}
Find a $2\times 2$ matrix function $N^{(1)}(y;z)$ that satisfies the following properties:

\begin{itemize}

 \item {} Analyticity: $N^{(1)}(z)$ is analytic in $\mathbb{C}\setminus((a,b)\cup(\frac{1}{b},\frac{1}{a})\cup(-b,-a)\cup(-\frac{1}{a},-\frac{1}{b}))$ and takes continuous boundary values on $(a,b)\cup(\frac{1}{b},\frac{1}{a})\cup(-b,-a)\cup(-\frac{1}{a},-\frac{1}{b})$(The directions of these open intervals are all facing upwards(see Figure \ref{fig3})).

 \item {} Jump condition: The boundary values on the jump contour $(a,b)\cup(\frac{1}{b},\frac{1}{a})\cup(-b,-a)\cup(-\frac{1}{a},-\frac{1}{b})$ are defined as
 \bee
N^{(1)}_{+}(z)=N^{(1)}_{-}(z)J_1(z),\quad z\in(a,b)\cup(\frac{1}{b},\frac{1}{a})\cup(-b,-a)\cup(-\frac{1}{a},-\frac{1}{b}),
\ene
where
\bee
J_1(z)
=\begin{cases}
\left[\!\!\begin{array}{cc}
1& 0  \vspace{0.05in}\\
-r_2(z)e^{-\frac{1}{2}(z+\frac{1}{z})y}& 1
\end{array}\!\!\right],\quad z\in(a,b),\vspace{0.05in}\\
\left[\!\!\begin{array}{cc}
1& 0  \vspace{0.05in}\\
-r_2(z)e^{-\frac{1}{2}(z+\frac{1}{z})y}& 1
\end{array}\!\!\right],\quad z\in(\frac{1}{b},\frac{1}{a}),\vspace{0.05in}\\
\left[\!\!\begin{array}{cc}
1& r_2(z)e^{\frac{1}{2}(z+\frac{1}{z})y}  \vspace{0.05in}\\
0& 1
\end{array}\!\!\right],\quad z\in(-b,-a),\v\\
\left[\!\!\begin{array}{cc}
1&  r_2(z)e^{\frac{1}{2}(z+\frac{1}{z})y}  \vspace{0.05in}\\
0& 1
\end{array}\!\!\right],\quad z\in(-\frac{1}{a},-\frac{1}{b}).
\end{cases}
\ene

 \item {} Normalization:
\bee
N^{(1)}(z)=\left\{\begin{array}{ll}
    \mathbb{I}_2+O\left(1/z\right)  & z\to\infty, \v\\
    \dfrac{(q+1)^2}{m^2+(q+1)^2}\left(\begin{array}{cc} 1&  \frac{im}{q+1} \v\\
\frac{im}{q+1} & 1 \end{array}\right)\left(\mathbb{I}_2+\mu_1^{(0)}(iz-i)\right)e^{\frac12 h_+\sigma_3} \v\\
 +\mathcal{O}((z-1)^2), & z\to 1.
\end{array}\right.
\ene

\end{itemize}
\end{prop}

According to Eq.~(\ref{fanyan-3}), we can recover $u(x)$ by the following formula:
\bee\label{fanyan-4}
u(x)=\lim\limits_{z\rightarrow 1}\frac{i}{1-z}\left(1-\dfrac{(N^{(1)}_{11}(z)+N^{(1)}_{21}(z))(N^{(1)}_{12}(z)+N^{(1)}_{22}(z))}{(N^{(1)}_{11}(1)+N^{(1)}_{21}(1))
(N^{(1)}_{12}(1)+N^{(1)}_{22}(1))}\right).
\ene
where
\bee\no
x(y)=y-\ln\left(\dfrac{N^{(1)}_{12}(1)+N^{(1)}_{22}(1)}{N^{(1)}_{11}(1)+N^{(1)}_{21}(1)}\right).
\ene

We set $\Sigma_1:=(a,b),\,\Sigma_3:=(\frac1b,\frac1a),\,\Sigma_5:=(-\frac1a,-\frac1b),\,\Sigma_7:=(-b,-a)$ and $\Sigma_2:=[\frac1a,a],\,\Sigma_4:=[-\frac1b,\frac1b],\,\Sigma_6:=[-a,-\frac1a]$. The directions of $\Sigma_j,j=1,\cdots,7$ are both to the right. To solve the Riemann-Hilbert problem \ref{RH5}, we make the following transformation:
\bee\label{T1}
N^{(2)}(y;z)=N^{(1)}(y;z)e^{yg(z)\sigma_3},
\ene
where $g(z)$ is scalar function to be determined by the following RH problem:

\begin{prop}\label{RH-g}
Find a scalar function $g(z)$ that satisfies the following properties:
\begin{itemize}

 \item {} Analyticity: $g(z)$ is analytic in $\mathbb{C}\setminus(-b,b)$ and takes continuous boundary values on $(-b,b)$;

 \item {} Jump condition: The boundary values on the jump contour are defined as
\begin{eqnarray}
&& \label{g}
g_+(z)+g_-(z)=\frac{1}{2}(z+\frac{1}{z}),\quad z\in \Sigma_1\cup\Sigma_3\cup\Sigma_5\cup\Sigma_7, \v\\
&& \label{g1}
g_+(z)-g_-(z)=m_1,\quad z\in \Sigma_2, \v\\
&& \label{g2}
g_+(z)-g_-(z)=m_2,\quad z\in \Sigma_4, \v\\
&&\label{g3}
g_+(z)-g_-(z)=m_3,\quad z\in \Sigma_6, \v\\
\end{eqnarray}
where $m_1,m_2,m_3$ are undetermined constants;

 \item {} Normalization:
 \bee\no
g(z)=\mathcal{O}(z^{-1}),\quad z\rightarrow\infty.
\ene

\end{itemize}
\end{prop}

To solve the scalar RH problem \ref{RH-g}, we first provide another RH problem about the $g$-derivative function.

\begin{prop}\label{RH-gd}
Find a scalar function $g'(z)$ that satisfies the following properties:
\begin{itemize}

 \item {} Analyticity: $g'(z)$ is analytic in $\mathbb{C}\setminus(-b,b)$ and takes continuous boundary values on $(-b,b)$;

 \item {} Jump condition: The boundary values on the jump contour are defined as
\begin{align}\label{g-d}
\begin{array}{ll}
g'_+(z)+g'_-(z)=\dfrac{1}{2}(1-\dfrac{1}{z^2}),& z\in \Sigma_1\cup\Sigma_3\cup\Sigma_5\cup\Sigma_7,\v\\
g'_+(z)-g'_-(z)=0, & z\in \Sigma_2\cup\Sigma_4\cup\Sigma_6;
\end{array}
\end{align}

 \item {} Normalization:
 \bee\label{g-d-N}
g'(z)=\mathcal{O}(z^{-2}),\quad z\rightarrow\infty.
\ene

\end{itemize}
\end{prop}

According to Eqs.~(\ref{g-d}) and (\ref{g-d-N}), we can construct the solution $g'(z)$ of the RH problem \ref{RH-gd} as follows:
\bee\label{g-d-1}
g'(z)=\frac{1}{4}\left(1-\frac{1}{z^2}-\frac{z^4+m_4z^2+m_5}{R(z)}\right),
\ene
where
\bee
R(z)=\sqrt{(z^2-\frac{1}{a^2})(z^2-\frac{1}{b^2})(z^2-a^2)(z^2-b^2)},
\ene
where $m_4,m_5$ are undetermined constants and the function $R(z)$ is positive on $(b, +\infty)$ with branch cuts on the contours $\Sigma_1\cup\Sigma_3\cup\Sigma_5\cup\Sigma_7$.

 Integrating Eq.~(\ref{g-d-1}), we have
\bee
g(z)=\frac{1}{4}\left(z+\frac{1}{z}-\int_{b}^{z}\frac{\zeta^4+m_4\zeta^2+m_5}{R(\zeta)}d\zeta\right).
\ene

\begin{lemma}\label{le4}
The parameters are determined as follows:
\begin{eqnarray}
&&m_4=\dfrac{\d\int_{\frac1a}^{a}\frac{\zeta^4}{R(\zeta)}d\zeta\int_{-\frac1b}^{\frac1b}\frac{1}{R(\zeta)}d\zeta
-\int_{\frac1a}^{a}\frac{1}{R(\zeta)}d\zeta\int_{-\frac1b}^{\frac1b}\frac{\zeta^4}{R(\zeta)}d\zeta}
{\d\int_{\frac1a}^{a}\frac{1}{R(\zeta)}d\zeta\int_{-\frac1b}^{\frac1b}\frac{\zeta^2}{R(\zeta)}d\zeta
-\int_{\frac1a}^{a}\frac{\zeta^2}{R(\zeta)}d\zeta\int_{-\frac1b}^{\frac1b}\frac{1}{R(\zeta)}d\zeta},\v\\ \label{m4} \\
&&m_5=\dfrac{\d\int_{\frac1a}^{a}\frac{\zeta^4}{R(\zeta)}d\zeta\int_{-\frac1b}^{\frac1b}\frac{\zeta^2}{R(\zeta)}d\zeta
-\int_{\frac1a}^{a}\frac{\zeta^2}{R(\zeta)}d\zeta\int_{-\frac1b}^{\frac1b}\frac{\zeta^4}{R(\zeta)}d\zeta}
{\d\int_{\frac1a}^{a}\frac{\zeta^2}{R(\zeta)}d\zeta\int_{-\frac1b}^{\frac1b}\frac{1}{R(\zeta)}d\zeta
-\int_{\frac1a}^{a}\frac{1}{R(\zeta)}d\zeta\int_{-\frac1b}^{\frac1b}\frac{\zeta^2}{R(\zeta)}d\zeta},\v\\ \label{m5} \\
&&m_1=\frac12\int_{a}^{b}\frac{\zeta^4+m_4\zeta^2+m_5}{R_+(\zeta)}d\zeta,\v\\ \label{m1} \\
&&m_2=\frac12\left(\int_{\frac1b}^{\frac1a}+\int_{a}^{b}\right)\frac{\zeta^4+m_4\zeta^2+m_5}{R_+(\zeta)}d\zeta,\v\\ \label{m2} \\
&&m_3=\frac12\left(\int_{-a}^{-\frac1a}+\int_{\frac1b}^{\frac1a}+\int_{a}^{b}\right)\frac{\zeta^4+m_4\zeta^2+m_5}{R_+(\zeta)}d\zeta.
\label{m3}
\end{eqnarray}

\end{lemma}

\begin{proof}
Eq.~(\ref{g}) implies that
\bee
\int_{\frac1a}^{a}\frac{\zeta^4+m_4\zeta^2+m_5}{R(\zeta)}d\zeta=0,\quad \int_{-\frac1b}^{\frac1b}\frac{\zeta^4+m_4\zeta^2+m_5}{R(\zeta)}d\zeta=0.
\ene

Then we obtain a system of linear equations about $m_4$ and $m_5$,
\begin{align}\label{le4-1g}
\begin{array}{l}
\d \int_{\frac1a}^{a}\frac{\zeta^4}{R(\zeta)}d\zeta+m_4\int_{\frac1a}^{a}\frac{\zeta^2}{R(\zeta)}d\zeta
+m_5\int_{\frac1a}^{a}\frac{1}{R(\zeta)}d\zeta=0,\vspace{0.15in}\\
\d \int_{-\frac1b}^{\frac1b}\frac{\zeta^4}{R(\zeta)}d\zeta+m_4\int_{-\frac1b}^{\frac1b}\frac{\zeta^2}{R(\zeta)}d\zeta+m_5\int_{-\frac1b}^{\frac1b}\frac{1}{R(\zeta)}d\zeta=0.
\end{array}
\end{align}

To solve system (\ref{le4-1g}), we can show that Eqs.~\eqref{m4}) and \eqref{m5} hold.
According to Eq.~(\ref{g1}), we obtain Eq.~\eqref{m1}.
According to Eqs.~(\ref{g2}) and (\ref{m1}), we obtain Eq.~\eqref{m2}.
Finally, Using Eqs.~(\ref{g3}),(\ref{m1}) and (\ref{m2}), we obtain Eq.~\eqref{m3}.
Thus the proof is completed.
\end{proof}

Therefore, the matrix function $N^{(2)}(y;z)$ satisfies the following Riemann-Hilbert problem.

\begin{prop}\label{RH7}
Find a $2\times 2$ matrix function $N^{(2)}(y;z)$ that satisfies the following properties:

\begin{itemize}

 \item {} Analyticity: $N^{(2)}(y;z)$ is analytic in $\mathbb{C}\setminus(-b,b)$ and takes continuous boundary values on $(-b,b)$(The directions of these open intervals are all facing upwards).

 \item {} Jump condition: The boundary values on the jump contour $(-b,b)$ are defined as
 \bee
N^{(2)}_{+}(y;z)=N^{(2)}_{-}(y;z)J_2(y;z),\quad z\in(-b,b),
\ene
where
\bee
J_2(y;z)
=\begin{cases}
\left[\!\!\begin{array}{cc}
e^{x(g_+(z)-g_-(z))}& 0  \vspace{0.05in}\\
-r_2(z)& e^{-x(g_+(z)-g_-(z))}
\end{array}\!\!\right],\quad z\in\Sigma_1\cup\Sigma_3,\vspace{0.05in}\\
\left[\!\!\begin{array}{cc}
e^{x(g_+(z)-g_-(z))}& r_2(z)  \vspace{0.05in}\\
0& e^{-x(g_+(z)-g_-(z))}
\end{array}\!\!\right],\quad z\in\Sigma_5\cup\Sigma_7,\v\\
\left[\!\!\begin{array}{cc}
e^{xm_1}&  0  \vspace{0.05in}\\
0& e^{-xm_1}
\end{array}\!\!\right],\quad z\in\Sigma_2,\v\\
\left[\!\!\begin{array}{cc}
e^{xm_2}&  0  \vspace{0.05in}\\
0& e^{-xm_2}
\end{array}\!\!\right],\quad z\in\Sigma_4,\v\\
\left[\!\!\begin{array}{cc}
e^{xm_3}&  0  \vspace{0.05in}\\
0& e^{-xm_3}
\end{array}\!\!\right],\quad z\in\Sigma_6;
\end{cases}
\ene

 \item {} Normalization:
\bee
N^{(2)}(y;z)=\left\{\begin{array}{ll}
    \mathbb{I}_2+O\left(1/z\right), & z\to\infty, \v\\
    \dfrac{(q+1)^2}{m^2+(q+1)^2}\left(\begin{array}{cc} 1&  \frac{im}{q+1} \v\\
\frac{im}{q+1} & 1 \end{array}\right)\left[\mathbb{I}_2+\mu_1^{(0)}(iz-i)\right]e^{(\frac12h_++yg(z))\sigma_3} \v\\
+\mathcal{O}((z-1)^2), & z\to 1.
\end{array}\right.
\ene

\end{itemize}
\end{prop}

Make the following transformation:
\bee\label{fz}
N^{(3)}(y;z)=N^{(2)}(y;z)f(z)^{\sigma_3},
\ene
where $f(z)$ is scalar function determined by the following RH problem:

\begin{prop} Find a scalar function $f(z)$ that satisfies the following properties:
\begin{itemize}

 \item {} Analyticity: $f(z)$ is analytic in $\mathbb{C}\setminus(-b,b)$ and takes continuous boundary values on $(-b,b)$;

 \item {} Jump condition: The boundary values on the jump contour $(-b,b)$ are defined as
\begin{eqnarray}
&& \label{f1}
f_+(z)f_-(z)=r_2^{-1}(z),\quad z\in \Sigma_1\cup\Sigma_3, \v\\
&& \label{f2}
f_+(z)f_-(z)=r_2(z),\quad z\in \Sigma_5\cup\Sigma_7, \v\\
&& \label{f3}
f_+(z)=f_-(z)e^{n_1},\quad z\in\Sigma_2\cup\Sigma_6, \v\\
&&\label{f4}
f_+(z)=f_-(z)e^{n_2},\quad z\in\Sigma_4,
\end{eqnarray}
where $n_1,n_2$ are undetermined constants;

 \item {} Normalization:
 \bee\label{f5}
f(z)=1+\mathcal{O}(z^{-1}),\quad z\rightarrow\infty;
\ene

\end{itemize}
\end{prop}

\begin{lemma}\label{le5}
The parameters $n_1, n_2$ are determined as follows:
\begin{align} \label{n12}
\begin{array}{l}
n_1=\dfrac{\d\int_{\Sigma_5\cup\Sigma_7}\dfrac{2\log r_2(\zeta)}{R_+(\zeta)}d\zeta\int_{\Sigma_4}\dfrac{\zeta^2}{R(\zeta)}d\zeta
-\int_{\Sigma_4}\dfrac{1}{R(\zeta)}d\zeta\int_{\Sigma_5\cup\Sigma_7}\dfrac{2\zeta^2\log r_2(\zeta)}{R_+(\zeta)}d\zeta}
{\d\int_{\Sigma_4}\dfrac{1}{R(\zeta)}d\zeta\int_{\Sigma_2}\dfrac{\zeta^2}{R(\zeta)}d\zeta-
\int_{\Sigma_2}\dfrac{2}{R(\zeta)}d\zeta\int_{\Sigma_4}\dfrac{\zeta^2}{R(\zeta)}d\zeta},\vspace{0.15in} \\
n_2=\dfrac{\d\int_{\Sigma_5\cup\Sigma_7}\dfrac{2\log r_2(\zeta)}{R_+(\zeta)}d\zeta\int_{\Sigma_2}\dfrac{\zeta^2}{R(\zeta)}d\zeta-
\int_{\Sigma_2}\dfrac{2}{R(\zeta)}d\zeta\int_{\Sigma_5\cup\Sigma_7}\dfrac{2\zeta^2\log r_2(\zeta)}{R_+(\zeta)}d\zeta}
{\d\int_{\Sigma_2}\dfrac{2}{R(\zeta)}d\zeta\int_{\Sigma_4}\dfrac{\zeta^2}{R(\zeta)}d\zeta-\int_{\Sigma_4}\dfrac{1}{R(\zeta)}d\zeta\int_{\Sigma_2}\dfrac{\zeta^2}{R(\zeta)}d\zeta}.
\end{array}
\end{align}

\end{lemma}

\begin{proof}
According to Eqs.~(\ref{f1})-(\ref{f4}), we can construct $f(z)$ as follows:
\bee
f(z)=\exp\left(\frac{R(z)}{2\pi i}F(z)\right),
\ene
with
\bee\label{Fz}
\begin{array}{rl}
F(z)=&\d \int_{\Sigma_1\cup\Sigma_3}\dfrac{\log r_2^{-1}(\zeta)}{R_+(\zeta)(\zeta-z)}d\zeta
+\int_{\Sigma_5\cup\Sigma_7}\dfrac{\log r_2(\zeta)}{R_+(\zeta)(\zeta-z)}d\zeta \v\v\\
& \d +\int_{\Sigma_2\cup\Sigma_6}\dfrac{n_1}{R(\zeta)(\zeta-z)}d\zeta+\int_{\Sigma_4}\dfrac{n_2}{R(\zeta)(\zeta-z)}d\zeta.
\end{array}
\ene

According to Eq.~(\ref{f5}), we know that the coefficients of the function $F(z)$ about $z^{-1},z^{-2},z^{-3},z^{-4}$ are all equal to zero, that is,
\begin{eqnarray}
&& \no
\d\int_{\Sigma_1\cup\Sigma_3}\dfrac{\log r_2^{-1}(\zeta)}{R_+(\zeta)}d\zeta+\int_{\Sigma_5\cup\Sigma_7}\dfrac{\log r_2(\zeta)}{R_+(\zeta)}d\zeta
+\int_{\Sigma_2\cup\Sigma_6}\dfrac{n_1}{R(\zeta)}d\zeta+\int_{\Sigma_4}\dfrac{n_2}{R(\zeta)}d\zeta=0, \vspace{0.15in} \\ \no\\
&&\no \d\int_{\Sigma_1\cup\Sigma_3}\dfrac{\zeta\log r_2^{-1}(\zeta)}{R_+(\zeta)}d\zeta+\int_{\Sigma_5\cup\Sigma_7}\dfrac{\zeta\log r_2(\zeta)}{R_+(\zeta)}d\zeta
+\int_{\Sigma_2\cup\Sigma_6}\dfrac{n_1\zeta}{R(\zeta)}d\zeta+\int_{\Sigma_4}\dfrac{n_2\zeta}{R(\zeta)}d\zeta=0,\vspace{0.15in}\\ \no\\
&&\no \d\int_{\Sigma_1\cup\Sigma_3}\dfrac{\zeta^2\log r_2^{-1}(\zeta)}{R_+(\zeta)}d\zeta+\int_{\Sigma_5\cup\Sigma_7}\dfrac{\zeta^2\log r_2(\zeta)}{R_+(\zeta)}d\zeta
+\int_{\Sigma_2\cup\Sigma_6}\dfrac{n_1\zeta^2}{R(\zeta)}d\zeta+\int_{\Sigma_4}\dfrac{n_2\zeta^2}{R(\zeta)}d\zeta=0,\vspace{0.15in}\\\no\\
&& \d\int_{\Sigma_1\cup\Sigma_3}\dfrac{\zeta^3\log r_2^{-1}(\zeta)}{R_+(\zeta)}d\zeta+\int_{\Sigma_5\cup\Sigma_7}\dfrac{\zeta^3\log r_2(\zeta)}{R_+(\zeta)}d\zeta
+\int_{\Sigma_2\cup\Sigma_6}\dfrac{n_1\zeta^3}{R(\zeta)}d\zeta+\int_{\Sigma_4}\dfrac{n_2\zeta^3}{R(\zeta)}d\zeta=0.\qquad
  \label{Fz-1}
\end{eqnarray}

Using the fact $r_2(z)=r_2(-z)$, Eqs.~(\ref{Fz-1}) can be rewrite as:
\begin{align}\label{Fz-2}
\begin{array}{l}
\d \int_{\Sigma_5\cup\Sigma_7}\dfrac{2\log r_2(\zeta)}{R_+(\zeta)}d\zeta
+n_1\int_{\Sigma_2}\dfrac{2}{R(\zeta)}d\zeta+n_2\int_{\Sigma_4}\dfrac{1}{R(\zeta)}d\zeta=0,\vspace{0.2in} \\
\d \int_{\Sigma_5\cup\Sigma_7}\dfrac{2\zeta^2\log r_2(\zeta)}{R_+(\zeta)}d\zeta
+n_1\int_{\Sigma_2}\dfrac{\zeta^2}{R(\zeta)}d\zeta+n_2\int_{\Sigma_4}\dfrac{\zeta^2}{R(\zeta)}d\zeta=0.
\end{array}
\end{align}

By solving system (\ref{Fz-2}), we can deduce Eq.~\eqref{n12}.
Thus the proof is completed.
\end{proof}

Thus, the matrix function $N^{(3)}(x;z)$ satisfies the following Riemann-Hilbert problem.

\begin{prop}\label{RH8}
Find a $2\times 2$ matrix function $N^{(3)}(y;z)$ that satisfies the following properties:

\begin{itemize}

 \item {} Analyticity: $N^{(3)}(y;z)$ is analytic in $\mathbb{C}\setminus(-b,b)$ and takes continuous boundary values on $(-b,b)$(The directions of these open intervals are all facing upwards).

 \item {} Jump condition: The boundary values on the jump contour $(-b,b)$ are defined as
 \bee
N^{(3)}_{+}(y;z)=N^{(3)}_{-}(y;z)J_3(y;z),\quad z\in(-b,b),
\ene
where
\bee
J_3(y;z)
=\begin{cases}
\left[\!\!\begin{array}{cc}
e^{y(g_+(z)-g_-(z))}\dfrac{f_+(z)}{f_-(z)}& 0  \vspace{0.05in}\\
-1& e^{-y(g_+(z)-g_-(z))}\dfrac{f_-(z)}{f_+(z)}
\end{array}\!\!\right],\quad z\in\Sigma_1\cup\Sigma_3,\vspace{0.05in}\\
\left[\!\!\begin{array}{cc}
e^{y(g_+(z)-g_-(z))}\dfrac{f_+(z)}{f_-(z)}& 1  \vspace{0.05in}\\
0& e^{-y(g_+(z)-g_-(z))}\dfrac{f_-(z)}{f_+(z)}
\end{array}\!\!\right],\quad z\in\Sigma_5\cup\Sigma_7,\v\\
\left[\!\!\begin{array}{cc}
e^{xm_1+n_1}&  0  \vspace{0.05in}\\
0& e^{-(xm_1+n_1)}
\end{array}\!\!\right],\quad z\in\Sigma_2,\v\\
\left[\!\!\begin{array}{cc}
e^{xm_2+n_2}&  0  \vspace{0.05in}\\
0& e^{-(xm_2+n_2)}
\end{array}\!\!\right],\quad z\in\Sigma_4,\v\\
\left[\!\!\begin{array}{cc}
e^{xm_3+n_1}&  0  \vspace{0.05in}\\
0& e^{-(xm_3+n_1)}
\end{array}\!\!\right],\quad z\in\Sigma_6;
\end{cases}
\ene

 \item {} Normalization:
\bee
N^{(3)}=\left\{\begin{array}{ll}
    \mathbb{I}_2+O\left(1/z\right),  & z\to\infty, \v\\
    \dfrac{(q+1)^2}{m^2+(q+1)^2}\left(\begin{array}{cc} 1&  \frac{im}{q+1} \v\\
\frac{im}{q+1} & 1 \end{array}\right)\left(\mathbb{I}_2+\mu_1^{(0)}(iz-i)\right)e^{(\frac12 h_++yg(z))\sigma_3}f(z)^{\sigma_3} & \v\\
\quad +\mathcal{O}((z-1)^2), & z\to 1.
\end{array}\right.
\ene

\end{itemize}
\end{prop}

\subsection{Opening lenses}

Let opening lenses $O_1$ pass through points $z=a$ and $z=b$ with boundary $O_1^{\pm}:=O_1\cap\mathbb{C}_{\pm}$ and opening lenses $O_2$ pass through points $z=\frac1b$ and $z=\frac1a$ with boundary $O_2^{\pm}:=O_2\cap\mathbb{C}_{\pm}$ and opening lenses $O_3$ pass through points $z=-\frac1a$ and $z=-\frac1b$ with boundary $O_3^{\pm}:=O_3\cap\mathbb{C}_{\pm}$ and opening lenses $O_4$ pass through points $z=-b$ and $z=-a$ with boundary $O_4^{\pm}:=O_4\cap\mathbb{C}_{\pm}$ (see Fig. \ref{fig1}). All boundaries $\partial O_j,j=1,2,3,4$ are in a clockwise direction. We define a new function as follows:
\bee
r_{3\pm}(z):=\pm r_2(z),\quad z\in\Sigma_1\cup\Sigma_3\cup\Sigma_5\cup\Sigma_7.
\ene

Note that the jump matrix $J_3(y;z)|_{\Sigma_1\cup\Sigma_3}$ has the following decomposition:
\begin{eqnarray}
&& \no J_3(y;z)=\left[\!\!\begin{array}{cc}
e^{y(g_+(z)-g_-(z))}\dfrac{f_+(z)}{f_-(z)}& 0  \vspace{0.05in}\\
-1& e^{-y(g_+(z)-g_-(z))}\dfrac{f_-(z)}{f_+(z)}
\end{array}\!\!\right]\vspace{0.05in}\\ \no\\
&&\no \qquad\quad\,\,\, =\left[\!\!\begin{array}{cc}
1& -\dfrac{e^{y(g_+(z)-g_-(z))}f_+(z)}{f_-(z)}  \vspace{0.05in}\\
0& 1
\end{array}\!\!\right]\left[\!\!\begin{array}{cc}
0& 1  \vspace{0.05in}\\
-1& 0
\end{array}\!\!\right]\left[\!\!\begin{array}{cc}
1& -\dfrac{e^{-y(g_+(z)-g_-(z))}f_-(z)}{f_+(z)}  \vspace{0.05in}\\
0& 1
\end{array}\!\!\right]\v\\ \no\\
&& \qquad\quad\,\,\, =\left[\!\!\begin{array}{cc}
1& \dfrac{e^{y(g_+(z)-g_-(z))}}{r_{3-}(z)f_-^2(z)}  \vspace{0.05in}\\
0& 1
\end{array}\!\!\right]\left[\!\!\begin{array}{cc}
0& 1  \vspace{0.05in}\\
-1& 0
\end{array}\!\!\right]\left[\!\!\begin{array}{cc}
1& -\dfrac{e^{-y(g_+(z)-g_-(z))}}{r_{3+}(z)f_+^2(z)}  \vspace{0.05in}\\
0& 1
\end{array}\!\!\right],\quad z\in \Sigma_1\cup\Sigma_3,\quad
\end{eqnarray}
and the jump matrix $J_3(x;z)|_{\Sigma_5\cup\Sigma_7}$ has the following decomposition:
\begin{align}
\begin{aligned}
J_3(x;z)&=\left[\!\!\begin{array}{cc}
e^{y(g_+(z)-g_-(z))}\dfrac{f_+(z)}{f_-(z)}& 1  \vspace{0.05in}\\
0& e^{-y(g_+(z)-g_-(z))}\dfrac{f_-(z)}{f_+(z)}
\end{array}\!\!\right]\vspace{0.05in}\\
&=\left[\!\!\begin{array}{cc}
1& 0  \vspace{0.05in}\\
\dfrac{e^{-y(g_+(z)-g_-(z))}f_-(z)}{f_+(z)}& 1
\end{array}\!\!\right]\left[\!\!\begin{array}{cc}
0& 1  \vspace{0.05in}\\
-1& 0
\end{array}\!\!\right]\left[\!\!\begin{array}{cc}
1& 0  \vspace{0.05in}\\
\dfrac{e^{y(g_+(z)-g_-(z))}f_+(z)}{f_-(z)}& 1
\end{array}\!\!\right]\v\\
&=\left[\!\!\begin{array}{cc}
1& 0  \vspace{0.05in}\\
-\dfrac{e^{-y(g_+(z)-g_-(z))}f_-^2(z)}{r_{3-}(z)}& 1
\end{array}\!\!\right]\left[\!\!\begin{array}{cc}
0& 1  \vspace{0.05in}\\
-1& 0
\end{array}\!\!\right]\left[\!\!\begin{array}{cc}
1& 0  \vspace{0.05in}\\
\dfrac{e^{y(g_+(z)-g_-(z))}f_+^2(z)}{r_{3+}(z)}& 1
\end{array}\!\!\right],\quad z\in\Sigma_5\cup\Sigma_7.
\end{aligned}
\end{align}

To eliminate the jumping of the characteristic function  on $\Sigma_5\cup\Sigma_7\cup\Sigma_5\cup\Sigma_7$, we make the following transformation:
\bee\label{N4}
N^{(4)}(x;z)=\begin{cases}
N^{(3)}\left[\!\!\begin{array}{cc}
1& \dfrac{e^{-y(2g(z)-\frac12(z+\frac1z))}}{r_{3}(z)f^2(z)}  \vspace{0.05in}\\
0& 1
\end{array}\!\!\right],\quad \mathrm{in~the~upper~lens}~O_1\cup O_2,\vspace{0.05in}\\
N^{(3)}\left[\!\!\begin{array}{cc}
1& \dfrac{e^{-y(2g(z)-\frac12(z+\frac1z))}}{r_{3}(z)f^2(z)}  \vspace{0.05in}\\
0& 1
\end{array}\!\!\right],\quad \mathrm{in~the~lower~lens}~O_1\cup O_2,\v\\
N^{(3)}\left[\!\!\begin{array}{cc}
1& 0  \vspace{0.05in}\\
-\dfrac{e^{y(2g(z)-\frac12(z+\frac1z))}f^2(z)}{r_{3}(z)}& 1
\end{array}\!\!\right],\quad \mathrm{in~the~upper~lens}~O_3\cup O_4,\vspace{0.05in}\\
N^{(3)}\left[\!\!\begin{array}{cc}
1& 0  \vspace{0.05in}\\
-\dfrac{e^{y(2g(z)-\frac12(z+\frac1z))}f^2(z)}{r_{3}(z)}& 1
\end{array}\!\!\right],\quad \mathrm{in~the~lower~lens}~O_3\cup O_4,\v\\
N^{(3)},\quad \mathrm{otherwise}.
\end{cases},
\ene

\begin{figure}[!t]
    \centering
 \vspace{-0.15in}
  {\scalebox{0.65}[0.6]{\includegraphics{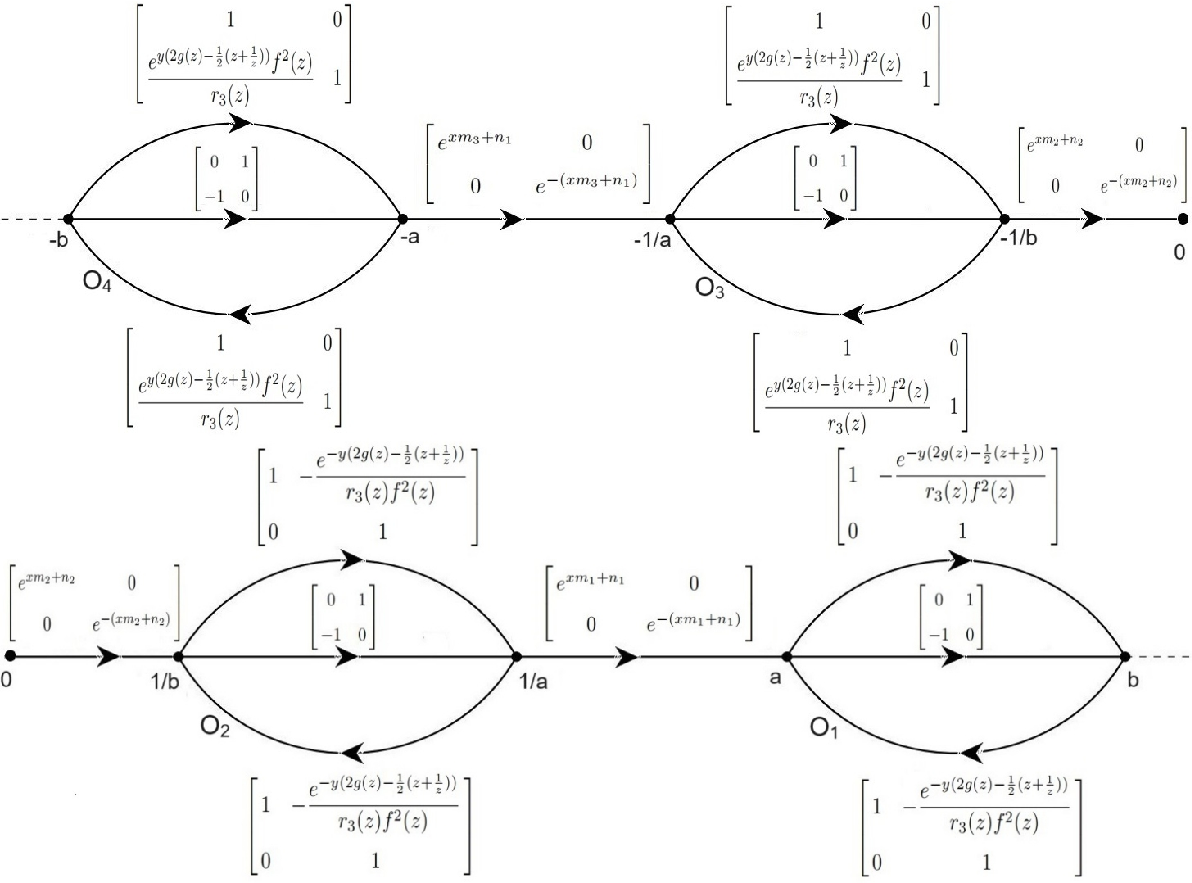}}}\hspace{-0.35in}
\vspace{0.1in}
\caption{The RH problem \ref{RH9} for the matrix function $N^{(4)}(y;z)$. Opening lenses $O_1$, $O_2$, $O_3$ and $O_4$.}
   \label{fig4}
\end{figure}

Then the matrix function $N^{(4)}(y;z)$ satisfies the following Riemann-Hilbert problem.

\begin{prop}\label{RH9}
Find a $2\times 2$ matrix function $N^{(4)}(y;z)$ that satisfies:

\begin{itemize}

 \item {} Analyticity: $N^{(4)}(y;z)$ is analytic in $\mathbb{C}\setminus(-b,b)$ and takes continuous boundary values on $(-b,b)$(The directions of these open intervals are all facing upwards(see Figure \ref{fig4})).

 \item {} Jump condition: The boundary values on the jump contour $(-b,b)$ are defined as
 \bee
N^{(4)}_{+}(y;z)=N^{(4)}_{-}(y;z)J_4(y;z),\quad z\in(-b,b),
\ene
where for $z\in O_j^{\pm},j=1,2,3,4$
\bee\label{J4}
J_4(y;z)
=\begin{cases}
\left[\!\!\begin{array}{cc}
1& -\dfrac{e^{-y(2g(z)-\frac12(z+\frac1z))}}{r_{3}(z)f^2(z)}  \vspace{0.05in}\\
0& 1
\end{array}\!\!\right],\quad z\in O_1^{\pm}\cup O_2^{\pm},\vspace{0.05in}\\
\left[\!\!\begin{array}{cc}
1& 0  \vspace{0.05in}\\
\dfrac{e^{y(2g(z)-\frac12(z+\frac1z))}f^2(z)}{r_{3}(z)}& 1
\end{array}\!\!\right],\quad z\in O_3^{\pm}\cup O_4^{\pm},
\end{cases}
\ene
and for $z\in (-b,b)$
\bee
J_4(y;z)
=\begin{cases} \,\,
i\sigma_2,\quad z\in \Sigma_1\cup\Sigma_3\cup\Sigma_5\cup\Sigma_7,\vspace{0.05in}\\
\left[\!\!\begin{array}{cc}
e^{xm_1+n_1}&  0  \vspace{0.05in}\\
0& e^{-(xm_1+n_1)}
\end{array}\!\!\right],\quad z\in\Sigma_2,\v\\
\left[\!\!\begin{array}{cc}
e^{xm_2+n_2}&  0  \vspace{0.05in}\\
0& e^{-(xm_2+n_2)}
\end{array}\!\!\right],\quad z\in\Sigma_4,\v\\
\left[\!\!\begin{array}{cc}
e^{xm_3+n_1}&  0  \vspace{0.05in}\\
0& e^{-(xm_3+n_1)}
\end{array}\!\!\right],\quad z\in\Sigma_6;
\end{cases}
\ene

 \item {} Normalization:
\bee
N^{(4)}(y;z)=\left\{\begin{array}{ll}
    \mathbb{I}_2+O\left(1/z\right),  & z\to\infty, \v\\
    \dfrac{(q+1)^2}{m^2+(q+1)^2}\left(\begin{array}{cc} 1&  \frac{im}{q+1} \v\\
\frac{im}{q+1} & 1 \end{array}\right)\left(\mathbb{I}_2+\mu_1^{(0)}(iz-i)\right)e^{(\frac12 h_++yg(z))\sigma_3}f(z)^{\sigma_3} & \v\\
\quad +\mathcal{O}((z-1)^2), & z\to 1.
\end{array}\right.\no
\ene
\end{itemize}
\end{prop}

\begin{lemma}\label{le6} The following inequalities hold:
\begin{align}
\begin{aligned}
&\mathrm{Re}(2g(z)-\frac12(z+\frac1z))<0,\quad z\in O_1^{\pm}\cup O_2^{\pm},\v\\
&\mathrm{Re}(2g(z)-\frac12(z+\frac1z))>0,\quad z\in O_3^{\pm}\cup O_4^{\pm}. \\
\end{aligned}
\end{align}
\end{lemma}
Notice that these two inequalities can not be  rigorously shown. But, we can check the results hold by using the numerical calculation.

According to Lemma \ref{le6}, we know that the off-diagonal entries of the jump matrix defined by Eq.~(\ref{J4}) along the
upper and lower lenses $O_1,O_2,O_3,O_4$ are exponentially decay when $y\to-\infty$.

\subsection{The outer parametrix}

We construct a new Riemann-Hilbert problem which only has the jump on $(-b,b)$.

\begin{prop}\label{RH10}
Find a $2\times 2$ matrix function $N^{o}(y;z)$ that satisfies the following properties:

\begin{itemize}

 \item {} Analyticity: $N^o(y;z)$ is analytic in $\mathbb{C}\setminus(-b,b)$ and takes continuous boundary values on $(-b,b)$.

 \item {} Jump condition: The boundary values on the jump contour are defined as
 \bee
N^o_{+}(y;z)=N^o_{-}(y;z)J^o(y;z),\quad z\in (-b,b),
\ene
where
\bee
J^o(y;z)
=\begin{cases}
\left[\!\!\begin{array}{cc}
0& 1  \vspace{0.05in}\\
-1& 0
\end{array}\!\!\right],\quad z\in \Sigma_1\cup\Sigma_3\cup\Sigma_5\cup\Sigma_7,\vspace{0.05in}\\
\left[\!\!\begin{array}{cc}
e^{xm_1+n_1}&  0  \vspace{0.05in}\\
0& e^{-(xm_1+n_1)}
\end{array}\!\!\right],\quad z\in\Sigma_2,\v\\
\left[\!\!\begin{array}{cc}
e^{xm_2+n_2}&  0  \vspace{0.05in}\\
0& e^{-(xm_2+n_2)}
\end{array}\!\!\right],\quad z\in\Sigma_4,\v\\
\left[\!\!\begin{array}{cc}
e^{xm_3+n_1}&  0  \vspace{0.05in}\\
0& e^{-(xm_3+n_1)}
\end{array}\!\!\right],\quad z\in\Sigma_6;
\end{cases}
\ene

 \item {} Normalization:
\bee
J^o(y;z)=\left\{\begin{array}{ll}
    \mathbb{I}_2+O\left(1/z\right),  & z\to\infty, \v\\
    \dfrac{(q+1)^2}{m^2+(q+1)^2}\left(\begin{array}{cc} 1&  \frac{im}{q+1} \v\\
\frac{im}{q+1} & 1 \end{array}\right)\left(\mathbb{I}_2+\mu_1^{(0)}(iz-i)\right)e^{(\frac12 h_++yg(z))\sigma_3}f(z)^{\sigma_3} & \v\\
\quad +\mathcal{O}((z-1)^2), & z\to 1.
\end{array}\right.
\ene
\end{itemize}
\end{prop}

According to the RH Problem 4.2.3 (Outer model problem) and Theorem 4.3.1 in Ref.~\cite{Kamvissis}, we can solve the RH problem \ref{RH10} by the Riemann theta function
\bee\no
\Theta(w)=\sum_{n\in\mathbb{Z}^G}\exp(\frac12n^\mathrm{T}Hn+n^\mathrm{T}w),\quad w\in\mathbb{C}^G,
\ene
where $G$ is an even nonnegative integer, and $H$ is a $G\times G$ matrix. The matrix function $N^{o}(y;z)$ is given by the formulae
\bee
\begin{array}{l}
\d N^{o}_{11}(y;z)=\dfrac{b^-(z)}{\beta(z)}\frac{\Theta(A^{\mathrm{cut}}(\infty)-V_1)}{\Theta(A^{\mathrm{cut}}(z)-V_1)}
\frac{\Theta(A^{\mathrm{cut}}(z)-V_1+iU/\hbar)}{\Theta(A^{\mathrm{cut}}(\infty)-V_1+iU/\hbar)},\v\v\\
\d N^{o}_{12}(y;z)=\dfrac{b^+(z)}{\beta(z)}e^{2ih_+(z_0)/\hbar}\frac{\Theta(A^{\mathrm{cut}}(\infty)-V_1)}{\Theta(-A^{\mathrm{cut}}(z)-V_1)}
\frac{\Theta(-A^{\mathrm{cut}}(z)-V_1+iU/\hbar)}{\Theta(A^{\mathrm{cut}}(\infty)-V_1+iU/\hbar)},\v\v\\
\d N^{o}_{21}(y;z)=\dfrac{b^+(z)}{\beta(z)}e^{-2ih_+(z_0)/\hbar}\frac{\Theta(-A^{\mathrm{cut}}(\infty)-V_2)}{\Theta(A^{\mathrm{cut}}(z)-V_2)}
\frac{\Theta(A^{\mathrm{cut}}(z)-V_2+iU/\hbar)}{\Theta(-A^{\mathrm{cut}}(\infty)-V_2+iU/\hbar)},\v\v\\
\d N^{o}_{11}(y;z)=\dfrac{b^-(z)}{\beta(z)}\frac{\Theta(-A^{\mathrm{cut}}(\infty)-V_2)}{\Theta(-A^{\mathrm{cut}}(z)-V_2)}
\frac{\Theta(-A^{\mathrm{cut}}(z)-V_2+iU/\hbar)}{\Theta(-A^{\mathrm{cut}}(\infty)-V_2+iU/\hbar)},
\end{array}
\ene
where the parameters are detailed in Ref.~\cite{Kamvissis}.

\subsection{The local parametrix }

In this section, we will construct a local matrix parametrix $N^{b}(y;z)$ as $z\in\delta^{b}:=\{z||z-b|<\epsilon,~\epsilon~\mathrm{is~a~suitable~small~ positive~parameter} \}$. Similarly, we can define regions $\delta^{a}$, $\delta^{-b},\delta^{-a},\delta^{\frac{1}{a}},\delta^{\frac{1}{b}},\delta^{-\frac{1}{b}}$, and $\delta^{-\frac{1}{a}}$. Note that
\bee
z+\frac1z-4g_{\pm}(z)=\mathcal{O}(\sqrt{z-b}),\quad z\to b.
\ene

Then, we define the following conformal map:
\bee
\zeta_b:=\dfrac{y^2(g(z)-\frac14(z+\frac1z))^2}{4},\quad z\in\delta^{b}.
\ene

Make the following transformation:
\bee
Y^{(1)}(y;z)
=\begin{cases}
N^{(4)}(y;z)\left(\dfrac{e^{\frac{i\pi}{4}}}{\sqrt{r_3(z)}f(z)}\right)^{\sigma_3}e^{-2\zeta_b^{\frac12}\sigma_3}\left[\!\!\begin{array}{cc}
0& -i \\
i& 0
\end{array}\!\!\right],\quad z\in\mathbb{C}_+\cap\delta^{b},\v\\
N^{(4)}(y;z)\left(\dfrac{e^{\frac{i\pi}{4}}}{\sqrt{-r_3(z)}f(z)}\right)^{\sigma_3}e^{-2\zeta_b^{\frac12}\sigma_3}\left[\!\!\begin{array}{cc}
0& -i \\
i& 0
\end{array}\!\!\right],\quad z\in\mathbb{C}_-\cap\delta^{b}.
\end{cases}
\ene

\begin{lemma} Matrix function $N^{(1)}(y;z)$ satisfies the following jump conditions:
\bee
Y_+^{(1)}(y;z)
=\begin{cases}
Y_-^{(1)}(y;z)\left[\!\!\begin{array}{cc}
1& 0 \\
-i& 1
\end{array}\!\!\right],\quad z\in\delta^{b}\cap\{\mathrm{upper~and~lower~lenses}\},\v\\
Y_-^{(1)}(y;z)\left[\!\!\begin{array}{cc}
0& i \\
i& 0
\end{array}\!\!\right],\quad z\in\delta^{b}\cap[0,+\infty).
\end{cases}
\ene

\end{lemma}

\begin{proof}
If $z\in\delta^{b}\cap[0,+\infty)$, we have

\begin{align}\no
\begin{array}{rl}
Y_+^{(1)}(y;z)&=N^{(4)}_+(y;z)\left(\dfrac{e^{\frac{i\pi}{4}}}{\sqrt{r_3(z)}f(z)}\right)^{\sigma_3}e^{-2\zeta_b^{\frac12}\sigma_3}\left[\!\!\begin{array}{cc}
0& -i \vspace{0.05in} \\
i& 0
\end{array}\!\!\right]  \vspace{0.15in}\\
&=N^{(4)}_-(y;z)\left[\!\!\begin{array}{cc}
0& 1  \vspace{0.05in}\\
-1& 0
\end{array}\!\!\right]\left(\dfrac{e^{\frac{i\pi}{4}}}{\sqrt{r_3(z)}f(z)}\right)^{\sigma_3}e^{-2\zeta_b^{\frac12}\sigma_3}\left[\!\!\begin{array}{cc}
0& -i \vspace{0.05in}\\
i& 0
\end{array}\!\!\right] \vspace{0.15in}\\
&=Y_-^{(1)}(y;z)\left[\!\!\begin{array}{cc}
0& -i \vspace{0.05in}\\
i& 0
\end{array}\!\!\right]e^{2\zeta_b^{\frac12}\sigma_3}\left(\dfrac{e^{\frac{i\pi}{4}}}{\sqrt{r_3(z)}f(z)}\right)^{-\sigma_3}\left[\!\!\begin{array}{cc}
0& 1  \vspace{0.05in}\\
-1& 0
\end{array}\!\!\right] \vspace{0.15in} \\
& \qquad \times
 \left(\dfrac{e^{\frac{i\pi}{4}}}{\sqrt{r_3(z)}f(z)}\right)^{\sigma_3}e^{-2\zeta_b^{\frac12}\sigma_3}\left[\!\!\begin{array}{cc}
0& -i \vspace{0.05in}\\
i& 0
\end{array}\!\!\right]  \vspace{0.15in} \\
&=Y_-^{(1)}(y;z)\left[\!\!\begin{array}{cc}
0& i \\
i& 0
\end{array}\!\!\right].
\end{array}
\end{align}

If $z\in\delta^{b}\cap\{\mathrm{upper~lens}\}$, we have
\begin{align}\no
\begin{array}{rl}
Y_+^{(1)}(y;z)&=N^{(4)}_+(y;z)\left(\dfrac{e^{\frac{i\pi}{4}}}{\sqrt{r_3(z)}f(z)}\right)^{\sigma_3}e^{-2\zeta_b^{\frac12}\sigma_3}\left[\!\!\begin{array}{cc}
0& -i \vspace{0.05in}\\
i& 0
\end{array}\!\!\right] \vspace{0.15in} \\
&=N^{(4)}_-(y;z)\left[\!\!\begin{array}{cc}
1& -\dfrac{e^{-2y(g(z)-\frac14(z+\frac1z))}}{r_3(z)f^2(z)}  \vspace{0.15in}\\
0& 1
\end{array}\!\!\right]\left(\dfrac{e^{\frac{i\pi}{4}}}{\sqrt{r_3(z)}f(z)}\right)^{\sigma_3}
e^{-2\zeta_b^{\frac12}\sigma_3}\left[\!\!\begin{array}{cc}
0& -i \vspace{0.05in}\\
i& 0
\end{array}\!\!\right]  \vspace{0.15in} \\
&=Y_-^{(1)}(y;z)\left[\!\!\begin{array}{cc}
0& -i \vspace{0.05in}\\
i& 0
\end{array}\!\!\right]e^{2\zeta_b^{\frac12}\sigma_3}\left(\dfrac{e^{\frac{i\pi}{4}}}{\sqrt{r_3(z)}f(z)}\right)^{-\sigma_3}\left[\!\!\begin{array}{cc}
1& -\dfrac{e^{-2y(g(z)-\frac14(z+\frac1z))}}{r_3(z)f^2(z)}  \vspace{0.05in}\\
0& 1
\end{array}\!\!\right] \vspace{0.15in} \\
&\quad \times \left(\dfrac{e^{\frac{i\pi}{4}}}{\sqrt{r_3(z)}f(z)}\right)^{\sigma_3}e^{-2\zeta_b^{\frac12}\sigma_3}\left[\!\!\begin{array}{cc}
0& -i \vspace{0.05in}\\
i& 0
\end{array}\!\!\right] \vspace{0.15in} \\
&=Y_-^{(1)}(y;z)\left[\!\!\begin{array}{cc}
1& 0 \vspace{0.05in}\\
-i& 1
\end{array}\!\!\right].
\end{array}
\end{align}

As the same way, we can obtain the jump condition of matrix function $Y^{(1)}(y;z)$ as $z\in\delta^{b}\cap\{\mathrm{lower~lens}\}$. Thus the proof is completed.
\end{proof}

\begin{figure}[!t]
    \centering
 \vspace{-0.15in}
  {\scalebox{0.2}[0.2]{\includegraphics{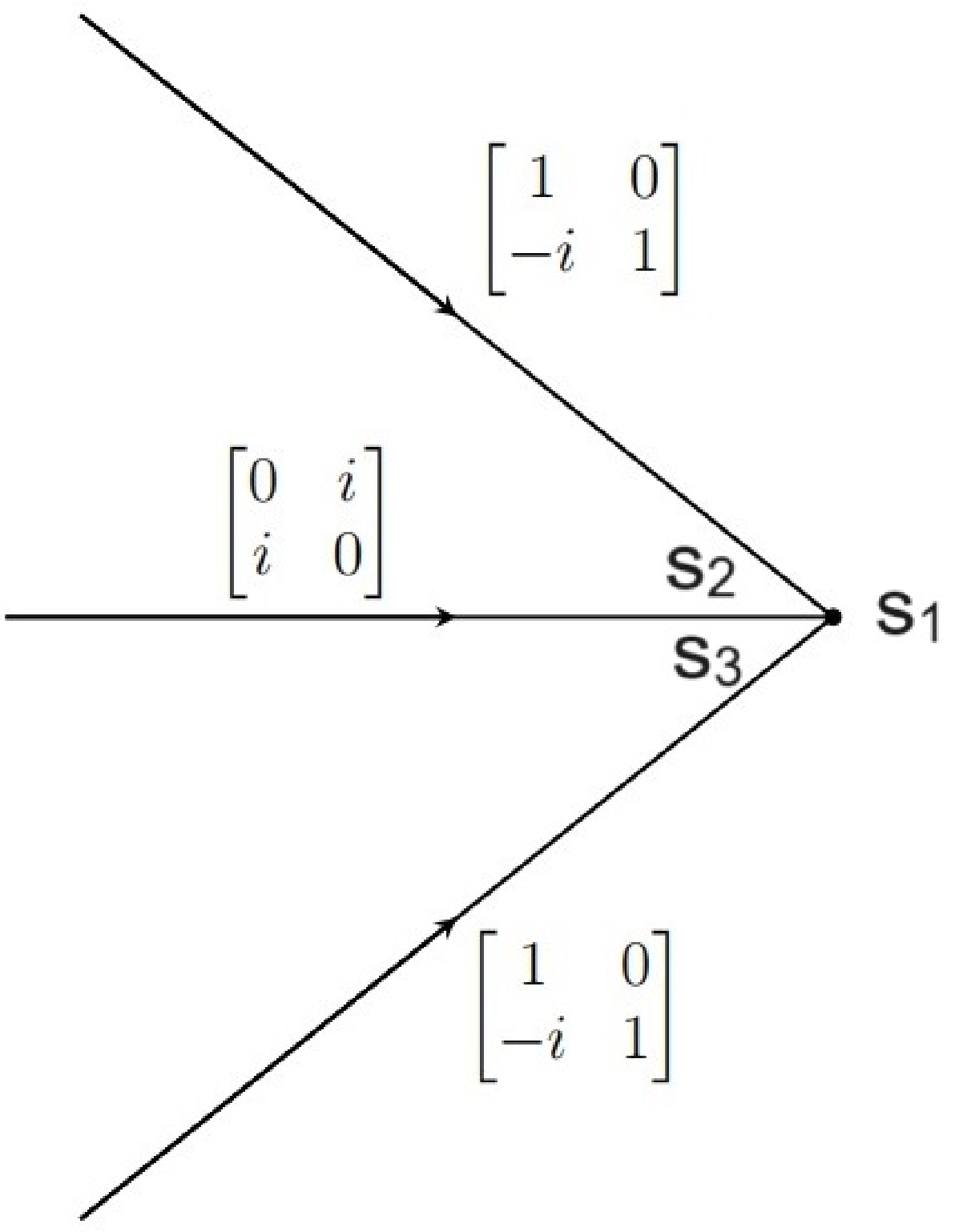}}}\hspace{-0.35in}
\vspace{0.05in}
\caption{Regional division for the Bessel model $U_{Bes}(z)$.}
   \label{fig2}
\end{figure}

Next, we will introduce the Bessel model $U_{Bes}(z)$ which satisfies the following Riemann-Hilbert problem:
\begin{prop}\label{BES}
Find a $2\times 2$ matrix function $U_{Bes}(z)$ that satisfies the following properties:

\begin{itemize}

 \item {} Analyticity: $U_{Bes}(z)$ is analytic for $z$ in the three regions shown in Figure \ref{fig2}, namely, $S_1: |\mathrm{arg}(z)|<\frac{2\pi}{3}$, $S_2: \frac{2\pi}{3}<\mathrm{arg}(z)<\pi$ and $S_3: -\pi<\mathrm{arg}(z)<-\frac{2\pi}{3}$ where $-\pi<\mathrm{arg}(z)\leq\pi$. It takes continuous boundary values on the excluded rays and at the origin from each sector.

 \item {} Jump condition: The boundary values on the jump contour $\gamma_U$ where $\gamma_U:=\gamma_{\pm}\cup\gamma_0$ with $\gamma_{\pm}=\{
 \mathrm{arg}(z)=\pm\frac{2\pi}{3}\}$ and $\gamma_0=\{
 \mathrm{arg}(z)=\pi\}$ are defined as
 \bee
U_{Bes+}(z)=U_{Bes-}(z)V_{B}(z),\quad z\in \gamma_U,
\ene
where
\bee
V_B(z)
=\begin{cases}
\left[\!\!\begin{array}{cc}
1& 0 \vspace{0.05in}\\
-i& 1
\end{array}\!\!\right],\quad z\in \gamma_{\pm},\vspace{0.05in}\\
\left[\!\!\begin{array}{cc}
0& i \vspace{0.05in}\\
i& 0
\end{array}\!\!\right],\quad z\in \gamma_0;
\end{cases}
\ene

 \item {} Normalization:
 \bee\no
U_{Bes}(z)=\left[\!\!\begin{array}{cc}
\mathcal{O}(\ln(|z|))& \mathcal{O}(\ln(|z|)) \vspace{0.05in}\\
\mathcal{O}(\ln(|z|))& \mathcal{O}(\ln(|z|))
\end{array}\!\!\right],\quad z\rightarrow 0.
\ene

\end{itemize}

\end{prop}

The solution of the Riemann-Hilbert Problem \ref{BES} can be expressed explicitly according to the Bessel functions~\cite{Kuijlaars-1}.  In particular, the solution satisfies:
\bee
U_{Bes}(z)=\frac{\sqrt{2}}{2}(2\pi z^{\frac12})^{-\frac12\sigma_3}\left[\!\!\begin{array}{cc}
1& -1 \vspace{0.05in}\\
1& 1
\end{array}\!\!\right]\left(\mathbb{I}+\mathcal{O}(z^{-\frac12})\right)e^{2z^{\frac12}\sigma_3}.
\ene
uniformly as $z\in\infty$ in the complex plane $\mathbb{C}$ aside from the jumps.

Then the local matrix parametrix $N^{b}(y;z)$ can be expressed as:
\bee
N^{b}(y;z)=N^{b1}(y;z)U_{Bes}(\zeta_b)\left[\!\!\begin{array}{cc}
0& -i \vspace{0.05in}\\
i& 0
\end{array}\!\!\right]e^{2\zeta_b^{\frac12}\sigma_3}\left(\dfrac{e^{\frac{i\pi}{4}}}{\sqrt{\pm r_3(z)}f(z)}\right)^{-\sigma_3},\quad z\in\delta^{b}\cap\mathbb{C_{\pm}}
\ene
where
\bee\no
N^{b1}(y;z)=\frac{\sqrt{2}}{2}N^o(y;z)\left(\dfrac{e^{\frac{i\pi}{4}}}{\sqrt{\pm r_3(z)}f(z)}\right)^{\sigma_3}\left[\!\!\begin{array}{cc}
i& -i \vspace{0.05in}\\
i& i
\end{array}\!\!\right](2\pi\zeta_b^{\frac12})^{\frac12\sigma_3}.
\ene

Then, we have
\bee
N^{b}(y;z)\left(N^o(y;z)\right)^{-1}=\mathbb{I}+\mathcal{O}(|y|^{-1}),\quad y\to-\infty,z\in\partial\delta^{b}\setminus\gamma_U.
\ene

Similarly, we can construct local matrix functions $N^{a}(y;z),N^{-a}(y;z)$ and $N^{-b}(y;z)$ which satisfy the following properties:
\begin{align}\no
\begin{aligned}
&N^{a}(y;z)\left(N^o(y;z)\right)^{-1}=\mathbb{I}+\mathcal{O}(|y|^{-1}),\quad y\to-\infty,z\in\partial\delta^{a}\setminus\gamma_U,\v\\
&N^{-a}(y;z)\left(N^o(y;z)\right)^{-1}=\mathbb{I}+\mathcal{O}(|y|^{-1}),\quad y\to-\infty,z\in\partial\delta^{-a}\setminus\gamma_U,\v\\
&N^{-b}(y;z)\left(N^o(y;z)\right)^{-1}=\mathbb{I}+\mathcal{O}(|y|^{-1}),\quad y\to-\infty,z\in\partial\delta^{-b}\setminus\gamma_U,\v\\
&N^{-\frac1a}(y;z)\left(N^o(y;z)\right)^{-1}=\mathbb{I}+\mathcal{O}(|y|^{-1}),\quad y\to-\infty,z\in\partial\delta^{-\frac1a}\setminus\gamma_U,\v\\
&N^{-\frac1b}(y;z)\left(N^o(y;z)\right)^{-1}=\mathbb{I}+\mathcal{O}(|y|^{-1}),\quad y\to-\infty,z\in\partial\delta^{-\frac1b}\setminus\gamma_U,\v\\
&N^{\frac1a}(y;z)\left(N^o(y;z)\right)^{-1}=\mathbb{I}+\mathcal{O}(|y|^{-1}),\quad y\to-\infty,z\in\partial\delta^{\frac1a}\setminus\gamma_U,\v\\
&N^{\frac1b}(y;z)\left(N^o(y;z)\right)^{-1}=\mathbb{I}+\mathcal{O}(|y|^{-1}),\quad y\to-\infty,z\in\partial\delta^{\frac1b}\setminus\gamma_U,
\end{aligned}
\end{align}

Then we construct a matrix function as follows:
\bee
Y_3(y;z)
=\begin{cases}
N^o(y;z),\quad z\in \mathbb{C}\setminus(\delta^{a}\cup\delta^{b}\cup\delta^{-a}\cup\delta^{-b}\cup\delta^{\frac1a}\cup\delta^{\frac1b}\cup\delta^{-\frac1a}\cup\delta^{-\frac1b}),\vspace{0.05in}\\
N^{a}(y;z),\quad z\in\delta^{a},\vspace{0.05in}\\
N^{b}(y;z),\quad z\in\delta^{b},\vspace{0.05in}\\
N^{-a}(y;z),\quad z\in\delta^{-a},\vspace{0.05in}\\
N^{-b}(y;z),\quad z\in\delta^{-b},\vspace{0.05in}\\
N^{-\frac1a}(y;z),\quad z\in\delta^{-\frac1a},\vspace{0.05in}\\
N^{-\frac1b}(y;z),\quad z\in\delta^{-\frac1b},\vspace{0.05in}\\
N^{\frac1a}(y;z),\quad z\in\delta^{\frac1a},\vspace{0.05in}\\
N^{\frac1b}(y;z),\quad z\in\delta^{\frac1b},
\end{cases}
\ene
which has the following jump condition:
 \bee
Y_{3+}(y;z)=Y_{3-}(y;z)J_5(y;z).
\ene

\subsection{Small-norm Riemann-Hilbert problem}

Define the error matrix function $E(y;z)$ as follow:
\bee
E(y;z)=N^{(4)}(y;z)Y_3^{-1}(y;z).
\ene
Then the error matrix function $E(y;z)$ satisfies the following Riemann-Hilbert problem.

\begin{prop}\label{RH11}
Find a $2\times 2$ matrix function $E(y;z)$ that satisfies the following properties:

\begin{itemize}

 \item {} Analyticity: $E(y;z)$ is analytic in $\mathbb{C}\setminus\gamma_3$ where $\gamma_3:=O_1\cup O_2\cup O_3\cup O_4\cup\partial\delta^{a}\cup\partial\delta^{b}\cup\partial\delta^{-a}\cup\partial\delta^{-b}\cup\partial\delta^{\frac1a}\cup\partial\delta^{\frac1b}\cup\partial\delta^{-\frac1a}\cup\partial\delta^{-\frac1b}
     \setminus(\delta^{a}\cup\delta^{b}\cup\delta^{-a}\cup\delta^{-b}\cup\delta^{\frac1a}\cup\delta^{\frac1b}\cup\delta^{-\frac1a}\cup\delta^{-\frac1b})$ and takes continuous boundary values on $\gamma_3$.

 \item {} Jump condition: The boundary values on the jump contour are defined as
 \bee
E_{+}(y;z)=E_{-}(y;z)V_E(y;z),\quad z\in \gamma_3,
\ene
where
\bee
V_E(y;z)=Y_{3-}(y;z)J_4(y;z)J_5^{-1}(y;z)Y_{3-}^{-1}(y;z),
\ene
which satisfies the following properties:
\bee\label{VE}
V_E
=\left\{\begin{array}{ll}
\mathbb{I}+\mathcal{O}(e^{-c|y|}), & z\in \cup_{j=1}^4O_j\setminus(\overline{\delta}^a\cup\overline{\delta}^b\cup\overline{\delta}^{-a}\cup\overline{\delta}^{-b}\cup
\overline{\delta}^{\frac1a}\cup\overline{\delta}^{\frac1b}\cup\overline{\delta}^{-\frac1a}\cup\overline{\delta}^{-\frac1b}),\v\\
\mathbb{I}+\mathcal{O}(\dfrac{1}{|y|}),& z\in\partial\delta^{a}\cup\partial\delta^{b}\cup\partial\delta^{-a}\cup\partial\delta^{-b}\cup\partial\delta^{\frac1a}\cup\partial\delta^{\frac1b}\cup\partial\delta^{-\frac1a}\cup\partial\delta^{-\frac1b},
\end{array}
\right.
\ene
with $c$ being a suitable positive real parameter.

 \item {} Normalization:
 \bee\no
E(y;z)=\mathbb{I}+\mathcal{O}(z^{-1}),\quad z\rightarrow\infty.
\ene

\end{itemize}
\end{prop}

Using Eq.~(\ref{VE}), we have
\bee
E_{+}(y;z)=E_{-}(y;z)\left(\mathbb{I}+\mathcal{O}(\frac{1}{|y|})\right),\quad z\in \gamma_3.
\ene

According to the standard theory of small-norm Riemann-Hilbert Problem, it follows that
\bee \label{Ey}
E(y;z)=\mathbb{I}+\mathcal{O}(\frac{1}{|y|}),\quad z\in \gamma_3.
\ene

\begin{theorem}
When $y\to-\infty$, the potential function $u(x)$ has the following asymptotic behavior:
\begin{align}
\begin{aligned}
u(x)=&\lim\limits_{z\rightarrow 1}\frac{i}{1-z}\left(1-\dfrac{(N^{(o)}_{11}(z)+N^{(o)}_{21}(z))(N^{(o)}_{12}(z)+N^{(o)}_{22}(z))}
{(N^{(o)}_{11}(1)+N^{(o)}_{21}(1))(N^{(o)}_{12}(1)+N^{(o)}_{22}(1))}\right)+\mathcal{O}(\frac{1}{|y|}),
\end{aligned}
\end{align}
where
\bee\label{xy}
x(y)=y-\ln\left(\dfrac{N^{(o)}_{12}(1)+N^{(o)}_{22}(1)}
{N^{(o)}_{11}(1)+N^{(o)}_{21}(1)}f^2(1)e^{2yg(1)}+\mathcal{O}(\frac{1}{|y|})\right).
\ene

\end{theorem}

\begin{proof}
Substituting Eqs.~(\ref{T1}),(\ref{fz}),(\ref{N4}) and (\ref{Ey}) into Eq.~(\ref{fanyan-4}), we have
\begin{align}\no
\begin{aligned}
u(x)=&\lim\limits_{z\rightarrow 1}\frac{i}{1-z}\left(1-\dfrac{(N^{(1)}_{11}(z)+N^{(1)}_{21}(z))(N^{(1)}_{12}(z)+N^{(1)}_{22}(z))}{(N^{(1)}_{11}(1)+N^{(1)}_{21}(1))
(N^{(1)}_{12}(1)+N^{(1)}_{22}(1))}\right)\\
=&
\lim\limits_{z\rightarrow 1}\frac{i}{1-z}\left(1-\dfrac{(N^{(o)}_{11}(z)f^{-1}(z)e^{-yg(z)}+N^{(o)}_{21}(z)f^{-1}(z)e^{-yg(z)})
}
{(N^{(o)}_{11}(1)f^{-1}(1)e^{-yg(1)}+N^{(o)}_{21}(1)f^{-1}(1)e^{-yg(1)})} \right. \\
& \left.\times \dfrac{(N^{(o)}_{12}(z)f(z)e^{yg(z)}+N^{(o)}_{22}(z)f(z)e^{yg(z)})}
{(N^{(o)}_{12}(1)f(1)e^{yg(1)}+N^{(o)}_{22}(1)f(1)e^{yg(1)})}
\right)+\mathcal{O}(\frac{1}{|y|}), \\
=& \lim\limits_{z\rightarrow 1}\frac{i}{1-z}\left(1-\dfrac{(N^{(o)}_{11}(z)+N^{(o)}_{21}(z))(N^{(o)}_{12}(z)+N^{(o)}_{22}(z))}
{(N^{(o)}_{11}(1)+N^{(o)}_{21}(1))(N^{(o)}_{12}(1)+N^{(o)}_{22}(1))}\right)+\mathcal{O}(\frac{1}{|y|}),
\end{aligned}
\end{align}
where $x(y)$ is given by Eq.\eqref{xy}.
\end{proof}

\section{$N_{\infty}$-soliton asymptotics: the elliptic domain}

In this section, we care about the $N_{\infty}$-asymptotic situation of $N$-soliton solution, under the additional assumptions:
\begin{itemize}

 \item {} The discrete spectrums $z_j,j=1,\cdots,N$ with the norming constants $c_j,j=1,\cdots,N$ fill uniformly compact domain $\Omega_2$, that is,
\bee
\Omega_2:=\{z|~\dfrac{(2y_1-a_1-a_2)^2}{4b_1^2}+\frac{x_1^2}{b_2^2}<1,~z=x_1+iy_1\},
\ene
where $ia_1$ and $ia_2$($a_2>a_1$) are the focal points of the ellipse $\partial\Omega_2$,  $b_1=\sqrt{b_2^2+(\frac{a_2-a_1}{2})^2}$, and $b_2$
is sufficiently small so that $\Omega_2\subset \{z\in\mathbb{C}|~0<\mathrm{arg}z<\pi,|z|>1\}$.

 \item {} The norming constants $c_j,j=1,\cdots,N$ have the following form:
\bee
c_j=\frac{|\Omega_2|r_4}{N\pi},
\ene
where $|\Omega_2|$ means the area of the domain $\Omega_2$ and $r_4(z)$ is a constant.

\end{itemize}

We define a closed curve $\Gamma_{5+}(\Gamma_{6+})$ with a very small radius encircling the poles $\{z_j\}_{j=1}^N(-\frac{1}{z_j})$ counterclockwise in the upper half plane $\mathbb{C}^+$, and a closed curve $\Gamma_{5-}(\Gamma_{6-})$ with a very small radius encircling the poles $\{-z_j\}_{j=1}^N(\frac{1}{z_j})$ counterclockwise in the lower half plane $\mathbb{C}^-$. According to Lemmas \ref{le1} and \ref{le2}, we obtain a Riemann-Hilbert problem $M_2(y,t;z):=\lim\limits_{N\to\infty}M^{(1)}(y,t;z)$.

\begin{prop}\label{RH12}
Find a $2\times 2$ matrix function $M_2(y,t;z)$ that satisfies the following properties:

\begin{itemize}

 \item {} Analyticity: $M_2(y,t;z)$ is analytic in $\mathbb{C}\setminus(\Gamma_{5\pm}\cup\Gamma_{6\pm})$ and takes continuous boundary values on $\Gamma_{5\pm}\cup\Gamma_{6\pm}$.

 \item {} Jump condition: The boundary values on the jump contour $\Gamma_{1+}\cup\Gamma_{1-}$ are defined as
 \bee
M_{2+}(y,t;z)=M_{2-}(y,t;z)V_4(y,t;z),\quad \lambda\in\Gamma_{5\pm}\cup\Gamma_{6\pm},
\ene
where
\bee\label{v4}
V_4(y,t;z)
=\begin{cases}
\left[\!\!\begin{array}{cc}
1& 0  \vspace{0.05in}\\
-\d\int_{\partial{\Omega_2}}\frac{\zeta^*r_4e^{-2i\theta(\zeta)}}{2\pi i(z-\zeta)}d\zeta& 1
\end{array}\!\!\right],\quad z\in\Gamma_{5+},\vspace{0.05in}\\
\left[\!\!\begin{array}{cc}
1& 0  \vspace{0.05in}\\
\d\int_{\partial{\Omega_2}}\frac{\zeta^*r_4e^{-2i\theta(-\frac{1}{\zeta})}}{2\pi i(z+\frac{1}{\zeta})\zeta^2}d\zeta& 1
\end{array}\!\!\right],\quad z\in\Gamma_{6+},\vspace{0.05in}\\
\left[\!\!\begin{array}{cc}
1&  \d\int_{\partial{\Omega_2}}\frac{\zeta^*r_4e^{2i\theta(-\zeta)}}{2\pi i(z+\zeta)}d\zeta \vspace{0.05in}\\
0& 1
\end{array}\!\!\right],\quad z\in\Gamma_{5-},\v\\
\left[\!\!\begin{array}{cc}
1& -\d\int_{\partial{\Omega_2}}\frac{\zeta^*r_4e^{2i\theta(\frac{1}{\zeta})}}{2\pi i(z-\frac{1}{\zeta})\zeta^2}d\zeta  \vspace{0.05in}\\
0& 1
\end{array}\!\!\right],\quad z\in\Gamma_{6-}.
\end{cases}
\ene

 \item {} Normalization:
\bee
 M_2(y, t; z)=\left\{\begin{array}{ll}
    \mathbb{I}_2+O\left(1/z\right)  & z\to\infty, \v\\
 \dfrac{(q+1)^2}{m^2+(q+1)^2}\left(\begin{array}{cc} 1&  \frac{im}{q+1} \v\\
\frac{im}{q+1} & 1 \end{array}\right)\left(\mathbb{I}_2+\mu_1^{(0)}(z-i)\right)e^{\frac12c_+\sigma_3} & \v\\
\quad +\mathcal{O}((z-i)^2), & z\to i.
\end{array}\right.
\ene
\end{itemize}
\end{prop}

\begin{lemma}\label{le4} The following identities hold:
\begin{eqnarray}
&& \no \d\int_{\partial{\Omega_2}}\frac{\zeta^*r_4e^{-2i\theta(\zeta)}}{2\pi i(z-\zeta)}d\zeta=\int_{ia_1}^{ia_2}\frac{\Delta F(\zeta)r_4e^{-2i\theta(\zeta)}}{2\pi i(z-\zeta)}d\zeta
,\v\v\\ \no\\
&& \no \d\int_{\partial{\Omega_2}}\frac{\zeta^*r_4e^{-2i\theta(-\frac{1}{\zeta})}}{2\pi i(z+\frac{1}{\zeta})\zeta^2}d\zeta=
\int_{ia_1}^{ia_2}\frac{\Delta F(\zeta)r_4e^{-2i\theta(-\frac{1}{\zeta})}}{2\pi i(z+\frac{1}{\zeta})\zeta^2}d\zeta,\v\v\\ \no\\
&& \no \d\int_{\partial{\Omega_2}}\frac{\zeta^*r_4e^{2i\theta(-\zeta)}}{2\pi i(z+\zeta)}d\zeta=
\int_{ia_1}^{ia_2}\frac{\Delta F(\zeta)r_4e^{2i\theta(-\zeta)}}{2\pi i(z+\zeta)}d\zeta,\v\v\\ \no\\
&& \d\int_{\partial{\Omega_2}}\frac{\zeta^*r_4e^{2i\theta(\frac{1}{\zeta})}}{2\pi i(z-\frac{1}{\zeta})\zeta^2}d\zeta=
\int_{ia_1}^{ia_2}\frac{\Delta F(\zeta)r_4e^{2i\theta(\frac{1}{\zeta})}}{2\pi i(z-\frac{1}{\zeta})\zeta^2}d\zeta,
\label{le4-1}
\end{eqnarray}
where $\Delta F(z)=-F_+(z)+F_-(z)$.
\end{lemma}
\begin{proof}

The boundary of $\Omega_2^*$, the complex conjugate domain of $\Omega_2$, is described by
\bee\label{le4-2}
z^*=(1-\frac{8b_1^2}{(a_2-a_1)^2})(z-\frac{i(a_1+a_2)}{2})+\frac{8b_1b_2}{(a_2-a_1)^2}F(z)-\frac{i(a_1+a_2)}{2},\quad z\in\partial\Omega_2,
\ene
where $F(z)=((z-ia_1)(z-ia_2))^{\frac12}$.
Using Eq.~(\ref{le4-2}), we can get Eq.~(\ref{le4-1}).
\end{proof}

According to Lemma \ref{le4}, we find that RH problem \ref{RH12} of matrix function $M_2(y,t;z)$ is equivalent to the case of line region.

\section{Conclusions and discussions}

We have explored the $N_\infty$-soliton asymptotic behaviors of the mCH equation with linear dispersion and zero boundaries.
We have chosen the following scattering data: the discrete spectra $K:=\{z_j,-z_j^*,\frac{1}{z_j^*},-\frac{1}{z_j}\}_{j=1}^N\cup K^*$, and norming constants
$\{c_j\}_{j=1}^N$, which correspond to discrete spectra. By constraining the discrete spectra $z_j$'s and theri corresponding norming constants $c_j$'s, we can obtain different types of $N_\infty$-soliton asymptotic behaviors. Case 1) The corresponding $N_\infty$-soliton asymptotic behavior is a one-soliton solution, where the discrete spectrum is located the center of the region when the region is a disk. Case 2) The corresponding $N_\infty$-soliton asymptotic behavior is an $n$-soliton solution when the region is a quadrature domain with $m=n$. This wave phenomenon, which represents a finite number of soliton interactions, is called soliton shielding. When the discrete spectra lie in the line region, we get its corresponding Riemann-Hilbert problem. When the discrete spectra lie in an ellipse region, it is equivalent to the case of the line region. These results on $N_\infty$-soliton asymptotic behaviors can provide a theoretical basis for related physical experiments.

Moreover, it is also important subject to
study the interactions of mulit-breather solutions with non-zero backgrounds. We will investigate
the $N_{\infty}$-breather asymptotics for the mCH equation with linear dispersion and non-zero backgrounds in future.

\vspace{0.2in}
\noindent {\bf Declaration of competing interest}

\vspace{0.05in}
The authors declare that they have no known competing financial interests or personal
relationships that could have appeared to influence the work reported in this paper.

\vspace{0.2in}
\noindent {\bf Data availability}

\vspace{0.05in}
All data generated or analyzed during this study are included in this published article
(and its supplementary information files).

\addcontentsline{toc}{section}{Acknowledgments}

\vspace{0.2in}
\noindent {\bf Acknowledgments}

\vspace{0.05in}
This work of W.W. was supported by the China Postdoctoral Science Foundation (No.043201025).
The work of Z.Y. was supported by the National Natural Science Foundation of China (No.11925108).


\addcontentsline{toc}{section}{References}

\begin{thebibliography}{100}\setlength{\itemsep}{-0.18mm}

{\small

\bibitem{Ablowitz2007} M. J. Ablowitz, P. A. Clarkson, {\it Solitons, Nonlinear Evolution Equations and Inverse Scattering} (Cambridge University Press, Cambridge, 1991).

\bibitem{Grava-3} M. Bertola, T. Grava, G. Orsatti, Soliton shielding of the focusing nonlinear Schr\"odinger equation, Phys. Rev. Lett. 130 (2023)  127201.


\bibitem{bil1} D. Bilman, R. Buckingham,  Large-order asymptotics for multiple-pole solitons of the focusing nonlinear Schr\"odinger equation. J. Nonlinear Sci. 29 (2019) 2185-2229.

\bibitem{bil3} D. Bilman, R. Buckingham, D. S. Wang, Far-field asymptotics for multiple-pole solitons in the large-order limit, J. Differ. Equ. 297 (2021) 320.

\bibitem{bil2} D. Bilman, L. Ling, P. D. Miller, Extreme superposition: rogue waves of infinite order and the Painlev\'e-III hierarchy, Duke Math. J. 169 (2020) 671-760.

\bibitem{RIST} D. Bilman, P. D. Miller, A robust inverse scattering transform for the focusing nonlinear Schr\"odinger equation, Comm. Pure Appl. Math. 72 (2019) 1792.







\bibitem{Bio} G. Biondini, D. Mantzavinos, Long-time asymptotics for the focusing nonlinear Schr\"odinger equation with nonzero boundary conditions at infinity and asymptotic stage of modulationalinstability, Comm. Pure Appl. Math. 70 (2017) 2300-2365.

\bibitem{NLS-18} M. Borghese, R. Jenkins, K.T.R. McLaughlin, P. Miller,  Long-time asymptotic behavior of the focusing nonlinear Schr\"odinger equation, Ann. Inst. Henri Poincar\'e, Anal. Non Lin\'eaire 35 (2018) 887-920.

\bibitem{BL1} A. Boutet de Monvel, J. Lenells, D. Shepelsky, Long-time asymptotics for the Degasperis-Procesi
equation on the half-line, Ann. Inst. Fourier (Grenoble) 69 (2019) 171-230.

\bibitem{BL2} A. Boutet de Monvel, J. Lenells, D. Shepelsky, The focusing NLS equation with step-like oscillating
background: scenarios of long-time asymptotics, Comm. Math. Phys. 383 (2021) 893-952.

\bibitem{BI} A. Boutet de Monvel, A. Its, V. Kotlyarov, Long-time asymptotics for the focusing NLS equation with time-periodic boundary condition on the half-line, Comm. Math. Phys. 290 (2009) 479-522.

  \bibitem{B1} A. Boutet de Monvel, I. Karpenko, D. Shepelsky, A Riemann-Hilbert approach to the modified
Camassa-Holm equation with nonzero boundary conditions, J. Math. Phys. 61 (2020) 031504.


\bibitem{B2} A. Boutet de Monvel, I. Karpenko, D. Shepelsky, The modified Camassa-Holm equation on a nonzero
background: large-time asymptotics for the Cauchy problem, Pure Appl. Funct. Anal. 7 (2022) 887-914; arXiv:2011.13235v1, 2020.

\bibitem{BS1} A. Boutet de Monvel, D. Shepelsky, Riemann-Hilbert approach for the Camassa-Holm equation on the line, C. R. Math. 343 (2006) 627-632.



\bibitem{BS2} A. Boutet de Monvel, D. Shepelsky, Riemann-Hilbert problem in the inverse scattering for the
Camassa-Holm equation on the line, Math. Sci. Res. Inst. Publ. 55 (2007) 53-75.


\bibitem{BS3} A. Boutet de Monvel, A. Kostenko, D. Shepelsky, and G. Teschl, Long-time asymptotics for the Camassa-Holm equation, SIAM J. Math. Anal. 41 (2009) 1559-1588.

\bibitem{CH} R. Camassa, D. Holm, An integrable shallow water equation with peaked solitons, Phys. Rev. Lett.
71 (1993) 1661-1664.

\bibitem{CH2} R. Camassa, D. D. Holm, J. M. Hyman, A new integrable shallow water equation, Adv. Appl. Mech. 31
(1994) 1-33.

\bibitem{Chang} X. Chang, J. Szmigielski, Lax integrability and the peakon problem
for the modified Camassa-Holm equation, Commun. Math. Phys. 358 (2018) 295-341.


\bibitem{CL} C. Charlier, J. Lenells, Boussinesq's equation for water waves: asymptotics in sector V, SIAM J.
Math. Anal. 56 (2024) DOI: 10.1137/23M1587671.

\bibitem{Chen15} R. Chen, Y. Liu, C. Qu, S. Zhang, Oscillation-induced blow-up to the modified Camassa-Holm equation with linear dispersion, Adv. Math. 272 (2015) 225-251.

\bibitem{E5} T. Congy, G. El, G. Roberti, Soliton gas in bidirectional dispersive hydrodynamics, Phys. Rev. E 103 (2021) 042201.

\bibitem{DF}  P. Deift, X. Zhou, A steepest descent method for oscillatory Riemann-Hilbert problems, Ann. Math.
137 (1993) 295-368.

\bibitem{Dbar3} M. Dieng, K.D.T. McLaughlin, Dispersive asymptotics for linear and integrable equations by the $\bar\partial$ steepest descent method, in: Nonlinear Dispersive Partial Differential Equations and Inverse Scattering, in: Fields Inst. Commun., vol. 83, Springer, New York, 2019, pp. 253-291.



\bibitem{E3} G.A. El, Critical density of a soliton gas, Chaos 26 (2016) 023105.

\bibitem{E1} G.A. El, A.M. Kamchantov, Kinetic equation for a dense soliton gas, Phys. Rev. Lett. 95 (2005) 204101.

\bibitem{E2} G.A. El, A.M. Kamchatnov, M.V. Pavlov, S.A. Zykov,  Kinetic equation for a soliton gas and
its hydrodynamic reductions, J. Nonlinear Sci. 21 (2011) 151-191.


\bibitem{E4} G.A. El, A. Tovbis, Spectral theory of soliton and breather gases for the focusing nonlinear
Schr\"odinger equation, Phys. Rev. E 101 (2020) 052207.

\bibitem{book-F}  L. D. Faddeev, L. A. Takhtajan, {\it Hamiltonian Methods in the Theory of Solitons} (Springer, Berlin, 1987).

\bibitem{fokas1} A.S. Fokas, The Korteweg-de Vries equation and beyond, Acta Appl. Math. 39 (1995) 295-305.

\bibitem{fokas} A.S. Fokas, On a class of physically important integrable equations, Physica D 87 (1995) 145-150.


\bibitem{Fuch} B. Fuchssteiner, Some tricks from the symmetry-toolbox for nonlinear equations: generalizations of
the Camassa-Holm equation, Physica D 95 (1996) 229-243.

\bibitem{FF} B. Fuchssteiner, A. Fokas, Symplectic structures, their B\"acklund transforms and hereditary symmetries, Physica D 4 (1981) 47-66.

    \bibitem{GGKM} C. S. Gardner, J.M. Greene, M.D. Kruskal, and R.M. Miura, Method for solving the Korteweg-de Vries equation, Phys. Rev Lett. 19 (1967) 1095-1097.

\bibitem{Girotti-1} M. Girotti, T. Grava, R. Jenkins, K. T. R. McLaughlin, Rigorous asymptotics of a KdV soliton gas, Commun. Math. Phys. 384 (2021) 733-784.

\bibitem{Girotti-2} M. Girotti, T. Grava, R. Jenkins, K. T. R. McLaughlin, A. Minakov, Soliton versus the gas: Fredholm determinants, analysis, and the rapid oscillations behind the kinetic equation, Comm. Pure Appl. Math. 76 (2023) 3233-3299.


\bibitem{GT} K. Grunert and G. Teschl, Long-time asymptotics for the Korteweg-de Vries equation vianonlinear steepest descent, Math. Phys. Anal. Geom. 12 (2009) 287-324.

\bibitem{Hou} Y. Hou, E. Fan, Z. Qiao, The algebro-geometric solutions for the Fokas-Olver-Rosenau-Qiao
(FORQ) hierarchy, J. Geom. Phys. 117 (2017) 105-133.

\bibitem{huang} L. Huang, J. Lenells, Asymptotics for the Sasa-Satsuma equation in terms of a modified Painlev\'e II
transcendent, J. Diff. Equa. 268 (2020) 7480-7504.

\bibitem{Its03} A. R. Its, The Riemann-Hilbert problem and integrable systems,   Notices of the AMS,  50 (2003) 1389-1400.

\bibitem{mch-dark}  R.I. Ivanov, T. Lyons, Dark solitons of the Qiao's hierarchy, J. Math. Phys. 53 (2012) 123701.


\bibitem{Kamvissis} S. Kamvissis, K. D. R. McLaughlin, P. D. Miller, Semiclassical Soliton Ensembles for the Focusing Nonlinear Schr\"odinger Equationv (Princeton University Press, 2003).

\bibitem{kang16}  J. Kang, X. Liu, P. J. Olver,  C. Qu, Liouville correspondence between the modified KdV
hierarchy and its dual integrable hierarchy, J. Nonlinear Sci. 26 (2016) 141-170.

\bibitem{kst}  I. Karpenko, D. Shepelsky, G. Teschl, A Riemann-Hilbert approach
to the modified Camssa-Holm equation with step-like boundary conditions,
arXiv:2203.05302v1.






\bibitem{Kuijlaars-1} A. B. J. Kuijlaars, K. R. McLaughlin, W. Van Assche, M. Vanlessen, The Riemann-Hilbert approach to strong asymptotics for orthogonal polynomials on [-1,1], Adv. Math. 188 (2004) 337-398.

\bibitem{LE} J. Lenells, The nonlinear steepest descent method: Asymptotics for initial-boundary value problems,
SIAM J. Math. Anal. 48 (2016) 2076-2118.

\bibitem{qu2}X. Liu, Y. Liu, C. Qu, Orbital stability of the train of peakons for an integrable modified
Camassa-Holm equation, Adv. Math. 255 (2014) 1-37.



\bibitem{Ma1} Y. Matsuno, Backlund transformation and smooth multisoliton solutions for a modified Camassa-Holm equation with cubic nonlinearity, J. Math. Phys. 54 (2013) 051504.

\bibitem{Ma2} Y. Matsuno, Smooth and singular multisoliton solutions of a modified Camassa-Holm equation with cubic nonlinearity and linear dispersion, J. Phys. A, Math. Theor. 47 (2014) 125203.



\bibitem{Dbar1} K.T.R. McLaughlin, P.D. Miller, The $\bar\partial$ steepest descent method and the asymptotic behavior of polynomials orthogonal on the unit circle with fixed and exponentially varying non-analytic weights, Int. Math. Res. Not. (2006) 48673.

\bibitem{Dbar2} K.T.R. McLaughlin, P.D. Miller, The $\bar\partial$ steepest descent method for orthogonal polynomials on the real line with varying weights, Int. Math. Res. Not. (2008) 075.


\bibitem{Zhou03} K.T.R. McLaughlin, X. Zhou (eds), Recent Developments in Integrable Systems and Riemann-Hilbert Problems (American Mathematical Society,  2003).

\bibitem{soliton-84} S. Novikov, S.V. Manakov, L.P. Pitaevskii, V.E. Zakharov, {\it Theory of Solitons
The Inverse Scattering Method} (Springer New York, 1984).


\bibitem{nov} V. Novikov, Generalizations of the Camassa-Holm type equation, J. Phys. A 42 (2009) 342002.


\bibitem{or96} P. J. Olver, P. Rosenau, Tri-Hamiltonian duality between solitons and solitary-wave solutions having
compact support, Phys. Rev. E 53 (1996) 1900-1906.



\bibitem{qiao06} Z. Qiao, A new integrable equation with cuspons and W/M-shape-peaks solitons, J. Math. Phys.
47 (2006) 112701.

\bibitem{qiao11} Z. Qiao, X.  Li, An integrable equation with nonsmooth solitons, Theor. Math. Phys. 267 (2011) 584-589.


\bibitem{qu1} C. Qu, X. Liu, Y. Liu, Stability of peakons for an integrable modified Camassa-Holm equation
with cubic nonlinearity, Commun. Math. Phys. 322 (2013) 967-997.

\bibitem{RS} Y. Rybalko, D. Shepelsky, Long-time asymptotics for the integrable nonlocal nonlinear Schr\"odinger equation, J. Math. Phys. 60 (2019) 031504.


\bibitem{SW} T. Sch\"afer, C.E. Wayne, Propagation of ultra-short optical pulses in cubic nonlinear media, Physica
D 196 (2004) 90-105.


\bibitem{sch96} J. Schiff, Zero curvature formulations of dual hierarchies, J. Math. Phys. 37 (1996) 1928-1938

\bibitem{E6} E. G. Shurgalina, E. N. Pelinovsky, Nonlinear dynamics of a soliton gas: Modified Korteweg-de Vries equation
framework, Phys. Lett. A 380 (2016) 2049-2053.

\bibitem{E7} P. Suret, A. Tikan, F. Bonnefoy, F. Copie, G. Ducrozet, A. Gelash, G. Prabhudesai, G. Michel, A. Cazaubiel, E. Falcon,
G. A. El, S. Randoux, Nonlinear spectral synthesis of soliton gas in deep-water surface gravity waves, Phys. Rev. Lett. 125 (2020) 264101.

\bibitem{TV} A. Tovbis, S. Venakides, X. Zhou, On the long-time limit of semiclassical (zero dispersionlimit) solutions of the focusing nonlinear Schr\"odinger equation: pure radiation case, Comm. Pure Appl. Math. 59 (2006) 1379-1432.

\bibitem{Olver16}T. Trogdon, S. Olver, {\it Riemann-Hilbert Problems, Their Numerical Solution and the Computation of Nonlinear Special Functions} (SIAM, 2016).

\bibitem{xu19} J. Xu, E. Fan, Long-time asymptotics behavior for the integrable modified Camassa-Holm equation with cubic nonlinearity, arXiv:1911.12554.

\bibitem{Yang10} J. Yang, {\it Nonlinear Waves in Integrable and Nonintegrable Systems} (SIAM, 2010).

\bibitem{fan-adv} Y. Yang, E. Fan, On the long-time asymptotics of the modified Camassa-Holm equation in space-time solitonic regions. Adv. Math. 402 (2022) 108340.


\bibitem{fan22} Y. Yang, G. Li, E. Fan, On the long-time asymptotics of the modified
Camassa-Holm equation with step-like initial data, arXiv:2203.10573v3.






\bibitem{zak71}  V.E. Zakharov, Kinetic equation for solitons. Sov. Phys. JETP 33 (1971) 538-541.


\bibitem{zs1} V.E. Zakharov, A.B. Shabat, A scheme for integrating the nonlinear equations of mathematical
physics by the method of the inverse scattering problem, Funkc. Anal. Prilozh. 6 (1974) 43-53.



\bibitem{zs2} V.E. Zakharov, A.B. Shabat, A scheme for integrating the nonlinear equations of mathematical
physics by the method of the inverse scattering problem. II, Funkc. Anal. Prilozh. 13 (1979) 13-22.


\bibitem{Zhou89} X. Zhou, The Riemann-Hilbert problem and inverse scattering, SIAM J. Math. Ana. 20  (1989) 966-986.




















     }

\end{thebibliography}
\end{document}